\documentclass[runningheads]{llncs}

\usepackage{amsmath}
\usepackage{amssymb}
\usepackage[utf8x]{inputenc}
\usepackage{listings}
\usepackage{color}
\usepackage{epsfig}
\usepackage{url}
\usepackage{amsmath}
\usepackage{cite}
\usepackage{graphicx}
\usepackage{color}
\usepackage{amsmath,dsfont}
\usepackage{epstopdf}
\usepackage{lipsum}
\usepackage{romannum}
\usepackage[ruled,linesnumbered, noend]{algorithm2e}
\usepackage{pgfplots}
\pgfplotsset{compat=1.13}
\usepackage{tikz}
\usepackage{graphicx}
\usepackage{breqn}

\usepackage{graphics}

\usepackage{csquotes}
\usepackage{amsmath}
\usepackage{subcaption}
\usepackage{graphicx}
\usepackage{accents}
\usepackage{mathtools}
\usepackage[normalem]{ulem}
\usepackage{dirtytalk}
\usepackage{mathrsfs}
\usepackage{amsmath}

\newcommand{\isdef}{\ensuremath{\overset{\scriptsize\textit{def}}{=}}}
\newcommand{\lp}{\ensuremath{\left(}}
\newcommand{\rp}{\ensuremath{\right)}}
\newcommand{\lb}{\ensuremath{\left[}}
\newcommand{\rb}{\ensuremath{\right]}}
\newcommand{\lc}{\ensuremath{\left\{}}
\newcommand{\rc}{\ensuremath{\right\}}}

\newcommand{\Nats}{\mathbb{N}}

\renewcommand{\P}{\mathds{P}}
\newcommand{\E}{\mathds{E}}

\newcommand{\ind}[1]{\mathds{1}_{#1}}

\newcommand{\rtot}{r^{\scriptsize\textrm{tot}}}
\newcommand{\ltot}{\ell^{\scriptsize\textrm{tot}}}

\newcommand{\sbbOne}{\textrm{SBB}}
\newcommand{\sbbTwo}{\textrm{S}^2\textrm{BB}}
\newcommand{\sbbThree}{\textrm{S}^3\textrm{BB}}

\newcommand{\ep}{\varepsilon^{\scriptsize\textrm{thm2}}}
\newcommand{\epdkw}{\varepsilon^{\scriptsize\textrm{dkw}}}

\usepackage{graphicx}
\begin{document}
\pagenumbering{arabic}
\pagestyle{plain}
\title{Quasi-Deterministic Burstiness Bound for Aggregate of Independent, Periodic Flows}
% \author{Seyed Mohammadhossein Tabatabaee\inst{1}\orcidID{0000-1111-2222-3333} \and
% Anne Bouillard\inst{2}\orcidID{1111-2222-3333-4444} \and
% Jean-Yves Le Boudec\inst{3}\orcidID{2222--3333-4444-5555}}
\author{Seyed Mohammadhossein Tabatabaee\inst{1} \and
Anne Bouillard\inst{2}\and
Jean-Yves Le Boudec\inst{3}}
\institute{EPFL, Lausanne, Switzerland, \email{hossein.tabatabaee@epfl.ch} \and
Huawei Technologies France, Paris, France, \email{anne.bouillard@huawei.com} \and
EPFL, Lausanne, Switzerland, \email{jean-yves.leboudec@epfl.ch}
}
\maketitle              % typeset the header of the contribution

\begin{abstract}
Time-sensitive networks require timely and accurate monitoring of the status of the network. To achieve this, many devices send packets periodically, which are then aggregated and forwarded to the controller. 
Bounding the aggregate burstiness of the traffic is then crucial for effective resource management. In this paper, we are interested in bounding this aggregate burstiness for independent and periodic flows. A deterministic bound is tight only when flows are perfectly synchronized, which is highly unlikely in practice and would be overly pessimistic. We compute the probability that the aggregate burstiness exceeds some value. When all flows have the same period and packet size, we obtain a closed-form bound using the Dvoretzky–Kiefer–Wolfowitz inequality. In the heterogeneous case, we group flows and combine the bounds obtained for each group using the convolution bound. Our bounds are numerically close to simulations and thus fairly tight. The resulting aggregate burstiness estimated for a non-zero violation probability is considerably smaller than the deterministic one: it grows in $\sqrt{n\log{n}}$, instead of $n$, where $n$ is the number of flows.

\end{abstract}
\section{Introduction} \label{sec:intro}

The development of industrial automation requires timely and accurate monitoring of the status of the network. In time-sensitive networks, a common assumption for critical types of traffic is that devices send packets periodically. These packets are aggregated and forwarded to the controller.  Characterizing this aggregate traffic is then crucial for effective resource management.  

Among the analytic tools providing analysis for real-time systems is deterministic network calculus~\cite{le_boudec_network_2001,bouillard_deterministic_2018}. From the characterization of the flows, the description of the switches (offered bandwidth and scheduling policy), it can derive worst-case performance bounds, such as end-to-end delay or buffer occupancy. These performances can grow linearly  with the burstiness of the flows~\cite{BN15}. Hence, accurately bounding the burstiness is key for performance evaluation and resource management. However, deterministic network calculus takes into account the worst-case scenario for aggregation of flows, which happens when flows are perfectly synchronized, and this is very unlikely to happen. 

To overcome this issue, probabilistic versions of network calculus (known as {\em Stochastic Network Calculus}) have emerged, and their aim is to compute performances when a small {\em violation probability} is allowed. Using probabilistic tools such as moment-generating functions~\cite{guide_snc} or martingales~\cite{PC14}, recent works mainly focus on ergodic systems and on the performances at an arbitrary point in time. This does not imply that the probability that the delay bound is never violated during a period of interest is small. 
Moreover, results are very limited in terms of topology and service policies, and become inaccurate for multiple servers~\cite{BNS22a}. 
Methods that compute probabilistic bounds on the burstiness have been discarded as they do not provide as good results for ergodic systems~\cite{sbb}. 

In this paper, we focus on the burstiness of the aggregation of periodic and independent flows. In other words, each flow sends packets periodically with a fixed packet size and period. The first packet  is sent at a random time (the phase) within that period. We assume phases are mutually independent. 
Our aim is to find a probabilistic bound on the burstiness of the aggregation of flows, that is, finding a burst that is valid {\em at all times} with large probability. That way, combining these probabilistic burstiness bounds with results of deterministic network calculus lead to delay and backlog bounds that are valid with large probability {\em at all times}, hence the name {\em quasi-deterministic}. To our knowledge, this is the first method to obtain quasi-deterministic bounds for independent, periodic flows.

Our contributions are the following: 
\begin{enumerate}
	\item %
 % First, in the homogeneous setting, where all flows have the same period and packet size, we provide two probabilistic bounds for the aggregate burstiness. Both of them are based on bounding the probability of some event $E$ relying on the order statistics of the phases. The former (Theorem~\ref{thm:homog_dkw}) has a closed form; it uses the Dvoretzky–Kiefer–Wolfowitz (DKW) inequality \cite{massart1990tight} to bound the probability of $E$, and allows to show that is a small positive violation probability is allowed, the burstiness grows in $O(\sqrt{n\log{n}})$, where $n$ is the number of flows, instead of $O(n)$ for a deterministic bound. The latter (Theorem~\ref{thm:homog}) directly computes directly the probability of $E$, which can be implemented iteratively. 
    First, in the homogeneous setting, where all flows have the same period and packet size, we provide two probabilistic bounds for the aggregate burstiness. Both of them are based on bounding the probability of some event $E$ relying on the order statistics of the phases. The former (Theorem~\ref{thm:homog_dkw}) has a closed form; it uses the Dvoretzky–Kiefer–Wolfowitz (DKW) inequality \cite{massart1990tight} to bound the probability of $E$, and when a small positive violation probability is allowed, the burstiness grows in $O(\sqrt{n\log{n}})$, where $n$ is the number of flows, instead of $O(n)$ for a deterministic bound. The latter (Theorem~\ref{thm:homog}) directly computes the probability of $E$, which can be implemented iteratively. 
	\item Second, we focus on two types of heterogeneity: either flows can be grouped into several homogeneous sub-sets, and we use a bounding technique based on convolution (Theorem~\ref{thm:hetero}), or flows have the same period but different packet sizes, and the bounds can be adapted from the homogeneous case (Theorem~\ref{thm:semi_homog}).
 \item Last, we numerically show that our bounds are close to simulations. The quasi-deterministic aggregate burstiness we obtain with a small, non-zero violation tolerance is considerably smaller than the deterministic one. For the heterogeneous case, we show that our convolution bounding technique provides bounds significantly  smaller than those obtained by the union bound. 
\end{enumerate}
The rest of the paper is organized as follows: We present our model in Section~\ref{sec:assump}, and then in Section~\ref{sec:realted_works} some results from state of the art. Our contributions are detailed in Section~\ref{sec:homogen} for the homogeneous case and Section~\ref{sec:heterogen} for the heterogeneous case. Finally, we provide some simulation results in Section~\ref{sec:num} to demonstrate the tightness of the bounds.

\section{Assumptions and Problem Statement} \label{sec:assump_prob}
% \hossein{
	% \begin{enumerate}
		%     \item Assumptions
		%     \item Discussion on stationarity.
		%     \item Discussion on ergodicity.
		%     \item Problem Statement
		% \end{enumerate}}
We use the notation $\Nats = \{0, 1, \ldots\}$ and $\Nats_n = \{1, \ldots, n\}$.
\paragraph{Assumptions}\label{sec:assump}

% We assume the phases are random and uniformly distributed in $[0,\tau_f]$, and that $(\phi_f)_{f\in \Nats_n}$ is a family of  mutually independent random variables. 

% We pick a time origin and let $\phi_f$ (phase) denote the emission time of the first packet. 

We consider $n$ periodic flows of packets. Each flow $f \in \Nats_n$  is periodic with period $\tau_f$ and phase $\phi_f \in [0, \tau_f)$, and sends packets  of size $\ell_f$: the number of bits of flow $f$ arriving in the time interval $[0, t)$ is $\ell_f\lceil [t-\phi_f]^+ / \tau_f\rceil$ where we use the notation $[x]^+ = \max(0,x)$ and $\lceil \rceil$ denotes the ceiling.

For every flow $f$, we assume that $\phi_f$ is random, uniformly distributed in $[0,\tau_f]$, and that the different $(\phi_f)_{f\in \Nats_n}$ are independent random variables.

% \begin{figure}[htbp]
	% \centering
	%     \scalebox{1}{\input{Figures/arrivals}}
	% 	\caption{\sffamily \small Traffic of flow $f$: Starting from phase $\phi_f$, flow $f$ periodically sends packets of constant size $l_f$  with period $\tau_f$. The traffic is constrained by a deterministic token-bucket arrival curve $\gamma_{r_f,l_f}$ defined for $t> 0$ by $\gamma_{r_f,l_f}(t)=l_f + r_ft$ and for $t=  0$ by $\gamma_{r_f,l_f}(0)=0$; note that $r_f = \frac{l_f}{\tau_f}$.}
	% 	\label{fig:periodic}
	% \end{figure}

\paragraph{Problem Statement} \label{sec:prob_state}
We consider the aggregation of the $n$ flows and let $A[s,t)$ denote the number of bits observed in time interval $[s,t)$. Our goal is to find a token-bucket arrival curve constraining this aggregate, that is, a rate $r$ and a burst $b$ such that $\forall s \leq t,~ A[s,t) \leq r(t - s) + b$. It follows from the assumptions that each individual flow $f\in\Nats_n$ is constrained  by a token-bucket arrival curve with rate $r_f= \ell_f / \tau_f$ and burst $\ell_f$.
%
% We say a flow is constrained by a token-bucket arrival curve with rate $r$ and burst $b$ if $\forall s \leq t,~ A(s,t] \leq r(t - s) + b$ where we let $A(s,t]$ denote the number of bits observed in time interval $(s,t]$. 
% the number of bits observed for this flow for any $s \leq t$ is upper bounded by $r(t-s) + b$, i.e., $\forall s \leq t,~ A(s,t] \leq r(t - s) + b$ where we let $A(s,t]$ denote the number of bits observed in time interval $(s,t]$. 
% % This definition is a deterministic constraint on the flow. 
% It follows, from the assumptions, that each flow $f=1:n$ is constrained  by a token-bucket arrival curve with rate $r_f= \frac{l_f}{\tau_f}$ and burst $l_f$.
% defined by $\gamma_{r_f,l_f}(t) = r_ft+l_f$ for $t > 0$ and  $\gamma_{r_f,l_f}(t) = 0$ for $t =0$ with $r_f = \frac{l_f}{\tau_f}$. 
% Also, the aggregate traffic has a deterministic, token-bucket arrival curve $\gamma_{\rtot,\ltot}$ with $\rtot = \sum_{f = 1}^Fr_f$ and $\ltot = \sum_{f = 1}^Fl_f$. 
Therefore, the aggregate flow is constrained by a token-bucket arrival curve with rate $\rtot = \sum_{f = 1}^nr_f$ and burst $\ltot = \sum_{f = 1}^n\ell_f$.
% , where $\ltot $ is the deterministic bound on the aggregate burstiness and is achieved when all phases are identical, i.e., when all flows are perfectly synchronized; this is highly unlikely in practice. In fact, the aggregate burstiness would likely be smaller when flows are not synchronized.
% In contrast, we assume that phases are independent, and we have no other information on them, which corresponds to the maximum entropy, thus, we assume uniform distribution.

However, due to the randomness of the phases, $\ltot$ might be larger than what is observed, and we are rather interested in token-bucket arrival curves  with rate $\rtot$ and a burst $b$ valid with some probability; specifically, we want to find a bound on the tail probability of the aggregate burstiness, which is defined as the smallest value of $B$ such that the aggregate flow is constrained by a token-bucket arrival curve with rate $\rtot$ and burst $B$, for the entire network lifetime. 
% bound the probability that the aggregate burstinss exceeds some value $b$ with a small probability $\epsilon(b)$, which we call a bound on the tail probability of the aggregate burstiness.
%
%We first  introduce an intermediate quantity $\bar{B}(t)$, which is the token-bucket content at time $t$ for a token-bucket that is initially empty, and is given by
% the aggregate burstiness $\bar{B}(t)$ at time $t$, which is defined as the smallest value of the burst such that the aggregate traffic is constrained by a token-bucket arrival curve with rate $\rtot$ and burst $\bar{B}(t)$, when observing the aggregate traffic up to time $t$. The value of  $\bar{B}(t)$ is given by
% We let $B(t)$ denote the aggregate burstiness at time $t$, and is given by 
%\begin{equation} \label{eq:agg_burst_t}
%	\bar{B}(t) = \sup_{s \leq t} \{ A[s,t) - \rtot(t - s) \}.
%\end{equation}
% where $A(s,t]$ is the number of bits observed in time interval $(s,t]$ for the aggregate traffic. 
The aggregate burstiness is given by 
\begin{equation}\label{eq:agg_burst}
	B = \sup_{t \geq 0}\Bar{B}(t).
\end{equation}
% and $B$ is the smallest value of the burst such that the aggregate traffic is constrained by a token-bucket arrival curve with rate $\rtot$ and burst $B$ throughout the entire network's lifetime. 
where
$\bar{B}(t)$ is the token-bucket content at time $t$ for a token-bucket that is initially empty, and is given by
% the aggregate burstiness $\bar{B}(t)$ at time $t$, which is defined as the smallest value of the burst such that the aggregate traffic is constrained by a token-bucket arrival curve with rate $\rtot$ and burst $\bar{B}(t)$, when observing the aggregate traffic up to time $t$. The value of  $\bar{B}(t)$ is given by
% We let $B(t)$ denote the aggregate burstiness at time $t$, and is given by 
\begin{equation} \label{eq:agg_burst_t}
	\bar{B}(t) = \sup_{s \leq t} \{ A[s,t) - \rtot(t - s) \}.
\end{equation}

Note that $B$ is a function of the random phases of the flows, therefore, is also random. Assume that $ \P(B > b) = \varepsilon$; this means that, with probability $1 - \varepsilon$, after periodic flows started, the aggregate burstiness is $\leq b$. Conversely, with probability $\varepsilon$, the aggregate burstiness is $> b$.

Observe that $ \P(B > b) =0$ for all $b \geq \ltot$, as $\ltot$ is a deterministic bound on the aggregate burstiness. Then, for some pre-specified value $0 \leq b < \ltot$, our problem is equivalent to finding $\epsilon(b) $ that bounds the tail probability of the aggregate burstiness $B$, i.e.,
\begin{equation} \label{eq:prob}
	\P(B > b) \leq \epsilon(b).
\end{equation}

\section{Background and Related Works} \label{sec:realted_works}

Bounding the burstiness of flows in Network Calculus is an important problem since it has a strong influence on the delay and backlog bounds. 
The deterministic aggregate burstiness can be improved (compared with summing burstiness of all flows)  when the phases of the flows are known exactly~\cite{DB16}. 

Regarding the \emph{stochastically bounded burstiness}~\cite{sbb, snc_basics}, three models have been proposed,
%The definition of a deterministic constraint on the flow is defined in Section~\ref{sec:assump_prob} and always hold. Alternatively, one can define an arrival constraint that holds with some probability. 
depending on how quantifiers are used, %called \emph{stochastically bounded burstiness}, \cite{sbb, snc_basics}, which we express here for token-bucket constraints: 
\begin{align}
	\label{eq:sbbOne}
	&\sbbOne: \forall 0 \leq s \leq t,~ \P\lp A[s, t) - r(t-s) > b\rp \leq \epsilon(b), \\
%	\label{eq:sbbTwo}
%	&\sbbTwo: \forall t \geq 0,~ \P ( \underbrace{\sup_{0 \leq s \leq t}\{A(t,s] - r(t-s)\}}_{\bar{B}(t)}  > b) \leq \epsilon(b) \\
%	\label{eq:sbbThree}
%	&\sbbThree: \P( \underbrace{\sup_{t\geq 0}\bar{B}(t)}_{B}  > b ) \leq \epsilon(b) 
	\label{eq:sbbTwo}
&\sbbTwo: \forall t \geq 0,~ \P (\sup_{0 \leq s \leq t}\{A[s, t) - r(t-s)\} > b) \leq \epsilon(b), \\
\label{eq:sbbThree}
&\sbbThree: \P(\sup_{t\geq 0} \{ \sup_{0 \leq s \leq t}\{  A[s, t) - r(t-s)\} \}  > b ) \leq \epsilon(b). 
	% &\sbbThree: \P( \sup_{t\geq 0}\bar{B}(t) = B  > b ) = \P( \exists t\geq 0, ~ B(t)  > b ) \leq \epsilon(b) 
\end{align}
First, notice that $\sbbThree \implies \sbbTwo \implies \sbbOne$. Indeed, $\sbbOne$ is a probability upper bound that the arrival curve constraint is invalid for a fixed pair of times $s\leq t$.   
In contrast, $\sbbTwo$ is the probability that token-bucket content at time $t$, $\bar{B}(t)$ exceeds $b$, hence the ``$\forall s$'' appearing inside the probability. Last, $\sbbThree$ represents the violation probability of the aggregate burstiness $B$ of the whole process. A deterministic arrival curve is a special case of $\sbbThree$, with $\epsilon(b) = 0$, which is why, for a non-zero violation probability $\epsilon(b)$, $b$ is called a \emph{quasi-deterministic} bound on the burstiness.

%\hossein{
%The first definition, $\sbbOne$, requires that, for a fixed pair of measurement instant $s \leq t$, the deterministic arrival curve constraint with rate $r$ and burst $b$ is satisfied with probability at least $1 - \epsilon(b)$, and this should hold for any $s \leq t$. The second definition, $\sbbTwo$, requires that for a fixed time $t$, with probability at least $1 - \epsilon(b)$, the deterministic arrival curve with rate $r$ and burst $b$ is satisfied for the entire interval $[0, t]$, and this should hold for any $t \geq 0$.
% upper bounds the probability that at a time $t$, the burstiness $B(t)$ exceeds $b$ (see \eqref{eq:agg_burst} for the definition of $B(t)$). 
%Lastly, $\sbbThree$, requires that, with probability at least $1 - \epsilon(b)$, the deterministic arrival curve with rate $r$ and burst $b$ is always satisfied.}

% upper bounds  the probability that there exists a time $t$ such that the burstiness $B(t)$ exceeds $b$.
%\hossein{
%Observe that $\sbbThree$ implies $\sbbTwo$  and $\sbbTwo$  implies $\sbbOne$, hence $\sbbThree$ is the strongest. Also,  a deterministic arrival curve is a special case of $\sbbThree$ where $\epsilon(b) = 0$. This is why, for a non-zero violation probability $\epsilon(b)$, $b$ is called a \emph{quasi-deterministic} bound on the burstiness.}

The first model $\sbbOne$ is the weakest, but also the easiest to handle: bounding the arrivals during a given  interval of time can be done for many stochastic models. It was also used for the study of aggregated independent flows with periodic patterns~\cite{milan_etal_backlog,chang_etal,guide_snc,sbb,stochastic_leaky_bucket_regulated_flows,kesidis_buffer}. All the approaches can be summarized as follows: a) defining an event $E_{s}$ of interest related to some time interval $[s, t)$ and aggregation of the flows; b) combining the events $(E_s)_{s\leq t}$ together to obtain a violation probability of the burstiness or of the backlog bound at time $t$. 

% \jylb{we should mention that S3BB $\Rightarrow$ S2BB $\Rightarrow$ SBB.}

The second model $\sbbTwo$ seems at first more adapted to network calculus analysis, as performance bounds can be directly derived from the formulation. However, the probability bound of $\sbbTwo$ is usually deduced from $\sbbOne$, which leads to pessimistic bounds for a single server. Nevertheless, this framework may become necessary for more complex cases~\cite{CBL06}.

In time-sensitive networks, we are interested in the probability that a delay or backlog bound is not violated during some interval (e.g., the network's lifetime), not just one arbitrary point in time, so the two models $\sbbOne$ and $\sbbTwo$ are not adapted, as they do not provide the violation probability of a delay bound  during a whole period of interest. 
In contrast, when using $\sbbThree$, we can guarantee, with some probability, that delay and backlog bounds derived by deterministic network calculus are never violated during the network's lifetime, which is why we choose this formulation in our model. 

% As explained in Section~\ref{sec:assump_prob},  we are interested in token-bucket arrival curves of type $\sbbThree$ rather than $\sbbTwo$: In time-sensitive networks, we are interested in the probability that a delay or backlog bound is not violated during some interval (e.g., the network's lifetime), not just one arbitrary point in time. however, bounds obtained by $\sbbTwo$ do not provide such information; the violation probability of a delay bound being small at one arbitrary point in time, does not imply that the probability that the delay bound is never violated during a period of interest is small. In fact, there would likely be some violations. In contrast, we can guarantee that delay bounds derived by deterministic network calculus, when using $\sbbThree$, would never be violated with some probability during the network's lifetime.

 As pointed out in~\cite[Section 4.4]{sbb}, when arrival processes are stationary and ergodic, $\sbbThree$ is always trivial and the bounding function $\epsilon(b)$ in \eqref{eq:sbbThree} is either zero or one. This is perhaps why the literature was discouraged from studying $\sbbThree$ characterizations. However, it has been overlooked that there is interest in some non-ergodic arrival processes, as in our case. Indeed, with our model, phases $\phi_f$ are drawn randomly but remain the same during the entire period of interest; thus, our arrival processes are not ergodic.

\section{Homogeneous Case} \label{sec:homogen}
% \section{Homogeneous Case: Flows with the Same Period and Same Packet-size} \label{sec:homogen}

In this section, we consider the case where flows have  the same packet size and same period. 

%\jylb{I would prefer to give the general result with arbitrary $\tau$ and $\ell$ and leave the simplification $\tau=\ell=1$ to the proof, for 2 reasons. First, it seems weird to have $\tau=\ell$ as the units are not the same. Second, and more important, we want our result to be used, and this simplifies the usage of our result.}

More precisely, we assume 
\begin{itemize}
	%\item[${\bf(H)}$] $\forall f\in\Nats_n,~\ell_f = \tau_f = 1$ and $(\phi_f)_{f\in\Nats_n}$ is a family of independent and identically distributed (iid)  uniform random variables (rv)  on $[0, 1)$.  
 \item[${\bf(H)}$] There exist $\tau, \ell>0$ such that $\forall f\in\Nats_n,~\ell_f = \ell, \tau_f=\tau$ and $(\phi_f)_{f\in\Nats_n}$ is a family of independent and identically distributed (iid)  uniform random variables (rv) on $[0, \tau)$.
\end{itemize}

We present two bounds for the aggregate burstiness; the former gives a closed form, unlike the latter, which might be slightly more accurate when the number of flows is small.

Let us first prove a useful result when the period $\tau$ is equal to $1$; it shows that if the time origin is shifted to the arrival time of the first packet of flow $i$, the phases of the $n-1$ other flows remain uniformly distributed on $[0, 1)$ and mutually independent. For this, we define the function $h$ as $\forall x, y\in[0, 1)$, 
\begin{equation}
	\label{eq:h}
	h(x, y) =  (x - y)\ind{x \geq y} + (1+x - y)\ind{x < y}.
\end{equation} 
Intuitively, if $x = \phi_j$ and $y = \phi_i$, $h(x, y)$ is the time until the arrival of the first packet of flow $j$, counted from the arrival time of the first packet of flow $i$.

\begin{lemma}\label{lem:uniform}
	Let $U_1,\ldots, U_{n}$ be a sequence of $n$ iid uniform rv on $[0, 1)$. 
	Let $i \in\Nats_n$ and	define $W_j$ for $j\in\Nats_{n} \setminus\{i\}$ by $W_j = h(U_j, U_i)$. 
%	\begin{equation}
%		W_j = (U_j - U_i)\ind{U_j \geq U_i} + (1+U_j - U_i)\ind{U_i<U_j}.
%	\end{equation}        
	% as follow: $j=1:(n-i),~ W_j = U_{(j+i)} - U_{(i)}$ and $ j=(n-i+1):(n-1),~ W_j = 1 + U_{(j - (n-i))} - U_{(i)}$. 
	Then, $(W_j)_{j\neq i}$ is a family of $n-1$ iid uniform rv on $[0, 1)$.
\end{lemma}

\begin{proof}
	Let us first do a preliminary computation for all $u_i\in [0, 1]$ and all bounded measurable function $g_j$: 
	\begin{align*}
		\E[g_j(h(U_j, u_i))] = \int_{u_j = 0}^1 \hspace{-0.2cm}g_j(h(u_j, u_i))du_j & = \int_{u_j = u_i}^1 \hspace{-0.2cm}g_j(u_j - u_i)du_j + \int_{u_j = 0}^{u_i}\hspace{-0.2cm} g_j(1 + u_j - u_i)du_j\\ = \int_{u_j = 0}^{1-u_i} g_j(w_j)dw_j &+ \int_{u_j=1-u_i}^1 g_j(w_j) dw_j = \int_{w_j = 0}^1 g_j(w_j)dw_j.
	\end{align*}
 Then, consider a collection of bounded measurable functions $(g_j)_{j\neq i}$: we can compute $\E[\prod_{j\neq i} g_j(W_j)] = \E[\prod_{j\neq i} g_j(h(U_j, U_i))] = \int_{u_i=0}^1 \E[\prod_{j\neq i} g_j(h(U_j ,u_i))] du_i  = \int_{0}^1 \prod_{j\neq i}\E[ g_j(h(U_j , u_i))] du_i = \int_{0}^1 \prod_{j\neq i}\E[ g_j(V_j)] du_i =  \prod_{j\neq i} \E[g_j(V_j)]=  \E[\prod_{j\neq i} g_j(V_j)]$, where $(V_j)_{j\neq j}$ is a collection of $n-1$ iid uniformly rv on $[0, 1)$. \qed
\end{proof}

The bounds we are to present are based on the order statistics: consider $n-1$ rv $U_1, \ldots, U_{n-1}$ and its order statistics is $U_{(1)}\leq  \cdots \leq U_{(n-1)}$, defined by sorting $U_1, \ldots, U_{n-1}$ in non-decreasing order. It is well-known \cite[Equation 1.145]{gentle2009computational} that if $(U_i)$ is an iid family of uniform rv on [0,1], the density function of the joint distribution of $U_{(1)},\ldots,U_{(n-1)}$ is
% \cite[Equation 1.145]{gentle2009computational}
\begin{equation}\label{eq:orderUniform}
	f_{U_{(1)},\ldots,U_{(n-1)}}( y_{1},\ldots,y_{n-1} )  = (n-1)!\ind{0 \leq y_{1}\leq y_2 \leq  \ldots \leq y_{n-1} \leq 1}.
\end{equation}

The next proposition connects the order statistics of the phases with the aggregate burstiness, and is key for Theorems~\ref{thm:homog_dkw} and~\ref{thm:homog}.

\begin{proposition}\label{prop:main}
	Assume model ${\bf (H).}$  For all $0\leq b < n \ell $, %Let $U_1, ..., U_{n-1}$ be a sequence of $n-1$ independent and identically distributed (iid) standard uniform random variables, and $U_{(1)},...,U_{(n-1)}$ be the corresponding order statistic, arranged in non-decreasing order. Then,
	\begin{align} \label{eq:BiT}
		% \forall i \in \{1,\ldots,n\}, 
		\P ( B > b) \leq n\P (E),
		% \lc B_i \leq b \rc = \bigcap_{k = 1}^{n-1} \lc U_{(k)} \geq \frac{ \lb (k + 1)l - b \rb^+}{\ltot} \rc, \;\;\; i = 1:n
	\end{align}
	with
	\begin{equation}
		E \isdef \bigcup_{k=\lfloor b / \ell\rfloor}^{n-1} 
		\lc U_{(k)} < \frac{(k+1)-b / \ell}{n}\rc,
		% E = \bigcup_{k=\lfloor \bar{b}\rfloor}^{n-1} 
		% \lc U_{(k)} < \frac{(k+1)l-b}{\ltot}\rc
		% = \bigcup_{k=\lfloor b\rfloor}^{n-1} 
		% \lc U_{(k)} < \frac{(k+1)-\Bar{b}}{n}\rc 
	\end{equation}
where $U_{(1)},\ldots,U_{(n-1)}$ is the order statistic of $n-1$ iid uniform rv on $[0,1]$.
\end{proposition}

\begin{proof}
Note that the normalized process $\tilde{A}[s, t) = \frac{1}{\ell} A[\tau s, \tau t)$ follows model ${\bf (H)}$ with $\tau = \ell = 1$, and 
$$\P(B > b) = \P(\tilde{B} > b/\ell).$$We then assume in this proof (and that of Theorem~\ref{thm:homog_dkw}) that $\tau = \ell = 1$, and the final result is obtained by replacing $b$ by $b/ \ell$. One can also remark that the bound is independent of $\tau$.

	Let $T_j,~ j\geq 1$ be the arrival time of the $j$-th packet in the aggregate. With probability 1, $T_j$ is strictly increasing  as we assume all phases are different.
	% We show the following equation:
	%  \begin{equation} \label{eq:s2}
		%     B  = \max_{i=1:n}  B_i
		% \end{equation}
	%     with 
	% \begin{align}
		% \label{eq:bi}
		%      % \forall i \in \{1,\ldots,n\}, ~ 
		%      B_i &\isdef \max_{j = i:(i+n-1)} H_{i,j}, \;\;\; i = 1:n\\
		%      \label{eq:Hij}
		%     \forall j \geq i, H_{i,j} &\isdef (j - i + 1)l - \ltot(T_j - T_i)
		% \end{align}
	%     \forall i \in \{1,2,\ldots, n\}, ~ B_i = \max_{j = i:(i+n-1)} \lc (j - i + 1)l - \ltot(T_j - T_i) \rc
	%     \forall j \geq i, H_{i,j} = (j - i + 1)l - \ltot(T_j - T_i)
	% \end{equation}
% We now prove Step 1: 
First, for all $i \leq j$, for all $(t_i, t_j)\in (T_{i-1},T_i]\times (T_j, T_{j+1}] \isdef C_{i,j}$, $A[t_i, t_j) = j-i+1$, and $H_{i, j} \isdef \sup_{t_i, t_j \in C_{i, j}}  A[t_i, t_j) - n(t_i-t_j) =  j-i+1 - n(T_j-T_i)$. Then, we can rewrite the aggregate burstiness as 
%$B = \sup_{1\leq i \leq j}  \sup_{t_i, t_j \in C_{i, j}}  A[t_i, t_j) - n(t_i-t_j) =\sup_{1\leq i \leq j}  H_{i, j}$. %, with $H_{i,j}   = \max_{i=1:j} \lc (j - i + 1)l - \ltot(T_j - T_i) \rc,~ j\leq i$
%
%As $\forall s \leq t, A[s, t)$ is constant between the arrival times of packets, observe that $B$ and $\bar{B}(t)$, defined in \eqref{eq:agg_burst} and \eqref{eq:agg_burst_t}, can be rewritten as 
%$B = \max_{1 \leq j}\bar{B}(T'_j)$ with $T_j < T'_j < T_{j+1}$ and $\bar{B}(T'_j) = \max_{i=1:j} \lc A[T_i, T'_j) - n(T_j - T_i) \rc$. Also,  observe that $A[T_i, T'_j) = j - i + 1$ and thus $\bar{B}(T_j)$ can be rewritten $\bar{B}(T_j)= \max_{i=1:j} H_{i,j}$ with $H_{i,j}   = \max_{i=1:j} \lc (j - i + 1)l - \ltot(T_j - T_i) \rc,~ j\leq i$; combine this with $B = \max_{1 \leq j}\bar{B}(T_j)$ and obtain 
\begin{equation} \label{eq:s1}
	%B = \max_{1 \leq j} \max_{i=1:j } H_{i,j} 
	B = \sup_{1\leq i \leq j}  \sup_{t_i, t_j \in C_{i, j}}  A[t_i, t_j) - n(t_i-t_j) =\sup_{1\leq i \leq j}  H_{i, j}.
\end{equation}
As our model is the aggregation of $n$ flows of period 1, $T_{j + n} = T_j + 1$ for $j\geq 1$, and $H_{i, j+n} =  j + n-i+1 - n(T_j - 1-T_i) = H_{i, j}$ for all $j \geq i$. Similarly, $H_{i+n,j} = H_{i,j}$ for all  $j \geq i + n$. Combine this with \eqref{eq:s1} and obtain
\begin{equation} \label{eq:p1s1}
	B =  \max_{j \geq 1} ~ \max_{i\in\Nats_j} H_{i,j} =  \max_{i\in\Nats_n} \underbrace{\max_{i\leq j \leq n-1} H_{i,j}}_{B_i}.
\end{equation}
% which concludes Step 1.
% As flows are packetized, points of interest are the arrival times of packets, and thus $\max_{t\geq 0}B(t) = \max_{1 \leq j}B(T_j)$ and $B(T_j) = \max_{i=1:j} \lc A(T_i, T_j] - \ltot(T_j - T_i) \rc$. Also, observe that $A(T_i, T_j] = (j - i + 1)l$ as in time interval $(T_i, T_j]$, $(j-i + 1)$ packets of size $l$ have arrived for the aggregate traffic. Hence
% \begin{equation} \label{eq:s1}
	%     B = \max_{1 \leq j} \max_{i=1:j } H_{i,j} 
	% \end{equation}
% As flows are periodic with a period equal to $1$, we have $T_{j + n} = T_j + 1$ for $j\geq 1$; also, recall that $\ltot = nl$. Then, it follows that $H_{i,j+n} = H_{i,j}$ for $i\leq j$ and $H_{i+n,j} = H_{i,j}$ if $i + n\leq j$. Hence,
% \begin{equation}
	%     \underbrace{\max_{1 \leq j} \max_{i=1:j } H_{i,j}}_{B} =  \max_{i=1:n} \underbrace{\max_{j = i:(i+n-1)} H_{i,j}}_{B_i}
	% \end{equation}
% which concludes \eqref{eq:s2}.

% Let $U_1, ..., U_{n-1}$ be a sequence of $n-1$ independent and identically distributed (iid) standard uniform random variables, and $U_{(1)},...,U_{(n-1)}$ be the corresponding order statistic. 
%We now prove the following:
%\begin{align} \label{eq:BiT}
%	% \forall i \in \{1,\ldots,n\}, 
%	\P \lp B_i > b \rp  = \P(E), \;\;\; i = 1:n
%\end{align}
We now prove that $\forall i \in \Nats_n$, $\P \lp B_i > b \rp  = \P(E)$. 

% \begin{align} \label{eq:BiT}
	%     % \forall i \in \{1,\ldots,n\}, 
	%     \lc B_i \leq b \rc = \bigcap_{k = 1}^{n-1} \lc U_{(k)} \geq \frac{ \lb (k + 1)l - b \rb^+}{\ltot} \rc, \;\;\; i = 1:n
	% \end{align}
Observe that for all $j\geq i$, we have the equality of events $\{H_{i, j} > b\} = \{T_j-T_i < (j-i+1-b) / n\}$, so for all $i \in \Nats_n$, $\lc B_i > b \rc = \bigcup_{j = i}^{i + n -1} \lc T_j - T_i < (j - i + 1 - b) / n \rc$.
% $\lc B_i > b \rc $ can be rewritten as follows: 
%   \begin{align} \label{eq:s21}
	%     \lc B_i > b \rc = \bigcup_{j = i}^{i + n -1} \lc (T_j - T_i) < \frac{(j - i + 1)l - b}{\ltot} \rc
	% \end{align}

We can also notice that the sequence $(T_j-T_i)_{j=i+1}^{n + i - 1}$ is the ordered sequence of phases starting from time origin $T_i$. Conditionally to $T_i = \phi_f$, or equivalently $\phi_{(i)} = f$, $(T_j-T_i)_{j=i+1}^{n + i - 1}$ is the order statistics of $(\phi_j - \phi_f)_{j\neq f}$, which is, from Lemma~\ref{lem:uniform}, iid and uniformly distributed on $[0, 1)$. If follows that 
$$\P(B_i > b~|~\phi_{(i)} = f) = \P(\cup_{k=1}^{n-1} \{U_{(k)} < \frac{k-b}{n}\}) = \P(\cup_{k=\lfloor b \rfloor}^{n-1} \{U_{(k)} < \frac{k-b}{n}\}) = \P(E),$$
since $U_{(k)} \geq 0$. Then, using the law of total probabilities, $\P(B_i > b) =  \sum_{f=1}^n \P(B_i > b~|~\phi_{(i)} = f)\P(\phi_{(i)} = f) = \P(E)$. 

Lastly, we conclude by using the union bound: $\P(B > b) = \P(\cup_{i=1}^n B_i>b) \leq \sum_{i=1}^n \P(B_i > b) = n\P(E)$.
\qed
\end{proof}

We now present the first bound on the tail probability of the aggregate burstiness $B$.

\begin{theorem}[Homogeneous case, DKW bound] \label{thm:homog_dkw}
Assume model ${\bf (H)}$ with $n>1$.  For all $b < n \ell$, a bound on the tail probability of the aggregate burstiness $B$ is given by 
\begin{equation}\label{eq:epsdkw}
	\P\lp B > b\rp \leq n\;\exp{\lp-2(n-1)\lp \frac{\lfloor b /\ell \rfloor}{n-1} - \frac{1}{n}\rp^2 \rp} \isdef \epdkw(n, \ell, b).
\end{equation}
% In the above, $B$ is the aggregate burstiness defined in \eqref{eq:agg_burst}.
% Also,  $\epdkw(n,l,b) \geq \ep(n,l,b)$ where $\ep(n,l,b)$ is the bound of Theorem~\ref{thm:homog}.
\end{theorem}
\begin{proof}
Let us assume that $\tau = \ell = 1$ in the proof, as in the proof of Proposition~\ref{prop:main}. 
Observe that when $ \lfloor b \rfloor < 1-\frac1n +\sqrt{\frac{(n-1)\log 2}{2}}$, we have $\epdkw(n, 1, b)\geq \frac n 2$, hence \eqref{eq:epsdkw} holds. Therefore we now proceed to prove \eqref{eq:epsdkw} when $ \lfloor b \rfloor \geq  1-\frac1n +\sqrt{\frac{(n-1)\log 2}{2}}$.

\textbf{Step 1:} Consider $n-1$ iid, rv $U_1, \ldots, U_{n-1}$ and its order statistics is $U_{(1)}\leq  \cdots \leq U_{(n-1)}$, defined by sorting $U_1, \ldots, U_{n-1}$ in non-decreasing order.  For $\varepsilon > 0$, define $E'(\varepsilon)$ by
\begin{equation} \label{eq:Eprime}
	E'(\varepsilon) \isdef \bigcup_{k=1}^{n-1} 
	\lc U_{(k)} < \frac{k}{n-1}-\varepsilon\rc. 
\end{equation}
We now show that if $\varepsilon \geq \sqrt{\frac{\log 2}{2(n-1)}} $, 
\begin{equation} 
	\P\lp E'(\varepsilon)\rp \leq e^{-2(n-1)\varepsilon^2}.
	\label{eq:s1DKW}
\end{equation} 
% \begin{equation} 
	%      \P\lp E'(\varepsilon)\rp \leq e^{-2(n-1)\varepsilon^2} \mif \varepsilon \geq \varepsilon_0 \melse    \P\lp E'(\varepsilon)\rp \leq \frac{1}{2}
	%      \label{eq:s1DKW}
	%  \end{equation} 
%   \begin{equation} 
	%     \P\lp E'(\varepsilon)\rp \leq e^{-2(n-1)\varepsilon^2} \ind{\epsilon \geq \epsilon_0} + \frac{1}{2}\ind{\epsilon < \epsilon_0} 
	%     % \mif \varepsilon \geq \sqrt{\frac{\log 2}{2(n-1)}} \melse    \P\lp E'(\varepsilon)\rp \leq \frac{1}{2}
	%     % \label{eq:s1DKW}
	% \end{equation} 
%       We show the following equation:
%   \begin{equation} 
	%      \P(E')\leq e^{-2(n-1)\varepsilon^2} \mif \varepsilon \geq \sqrt{\frac{\log 2}{2(n-1)}}
	%      \label{eq:s1DKW}
	%  \end{equation} 
%  with 
%   \begin{equation} \label{eq:Eprime}
	%    E' = \bigcup_{k=1}^{n-1} 
	%    \lc U_{(k)} < \frac{k}{n-1}-\varepsilon\rc  
	% \end{equation}
Let $F_{n-1}$ be the (random) empirical cumulative distribution function of $U_1,\ldots,U_{n-1}$, defined  $\forall x \in [0,1]$ by
\begin{equation}
	F_{n-1}(x)=\frac{1}{n-1}\sum_{i=1}^{n-1}\ind{U_{(i)}\leq x}.
\end{equation} 
The Dvoretzky–Kiefer–Wolfowitz inequality \cite{massart1990tight} states that if $\varepsilon \geq \sqrt{\frac{\log 2}{2(n-1)}}$, then 
\begin{equation}
\label{eq-wpjgrt}
	\P \big(\sup_{x\in[0,~1]} (F_{n-1}(x)-x)>\varepsilon  \big) \leq e^{-2(n-1)\varepsilon^2}.
\end{equation} 
We can apply this to find the bound of interest. First, we prove that
\begin{equation} \label{hgqewd}
	\sup_{x\in[0,~1]} (F_{n-1}(x)-x)>\varepsilon  
	\Leftrightarrow
	\exists k\in \Nats_{n-1},~U_{(k)}<\frac{k}{n-1}-\varepsilon.
\end{equation}

Proof of  $\Leftarrow$: First, observe that $F_{n-1}(U_{(k)}) = k / (n-1)$, so if  $\frac{k}{n-1} - U_{(k)} > \varepsilon$ for some $k$, then  $F_{n-1}\lp U_{(k)} \rp - U_{(k)}>\varepsilon$, and the left-hand side holds. %with $x = U_{(k)}$.

Proof of $\Rightarrow$: Set $U_{(0)} = 0$ and $U_{(n)} = 1$. Observe that  for all $k \in \{0, \ldots, n-1\}$,  and all $U_{(k)} \leq x < U_{(k+1)}$, $F_{n-1}\lp x\rp = F_{n-1}\lp U_{(k)}\rp = \frac{k}{n-1}$. 
%Also, observe that $x > U_{(n-1)}$, $F_{n-1}\lp x\rp = 1$ and $x < U_{(1)}$, $F_{n-1}\lp x\rp = 0$. 
Hence, $F_{n-1}(x) - x = \frac{k}{n-1} - x$ is decreasing on each segment $[U_{(k)}, U_{(k+1)})$.
Then,
 %\begin{equation}
%	F_{n-1}(x) - x =  \begin{cases} 
%		-x & \text{if } 0 \leq x < U_{(1)}\\
%		\frac{k}{n-1} - x & \text{if } U_{(k)} \leq x < U_{(k+1)} \text{for $k =1:(n-2)$} \\
%		1 - x & \text{if }   U_{(n-1)} \leq x \leq  1
%	\end{cases}
%\end{equation}
%Thus, $F_{n-1}(x) - x \leq 0 \leq \epsilon$ when $0 \leq x < U_{(1)}$  and $F_{n-1}(x) - x$ is decreasing for $U_{(n)} \leq x \leq  1$ and for $ U_{(k)} \leq x < U_{(k+1)}$ with $k =1:(n-2)$.
the supremum in the left-hand side of \eqref{hgqewd} is obtained for some $x =  U_{(k)}$, i.e., $\sup_{x\in[0,~1]} (F_{n-1}(x)-x) = \sup_{k\in\{ 0,\ldots,n-1\}} (F_{n-1}(U_{(k)})-U_{(k)}) = \sup_{k\in\{ 0,\ldots,n-1\}} (\frac{k}{n-1}-U_{(k)})$, which implies the right-hand side ($F_{n-1}(0) - 0 = 0 < \varepsilon$).

% it implies that $\exists k \in \{ 1,\ldots,n-1\}$ the supremum is obtained for some $x =  U_{(k)}$, i.e., 

% First, observe that for $k = 1:(n-2)$ and $U_{(k)} \leq x < U_{(k+1)}$, $F_{n-1}\lp x\rp = F_{n-1}\lp U_{(k)}\rp = \frac{k}{n-1}$. Second, observe that $x > U_{(n-1)}$, $F_{n-1}\lp x\rp = 1$ and $x < U_{(1)}$, $F_{n-1}\lp x\rp = 0$. Third, observe

% , and as $F_{n-1}\lp U_{(k)}\rp =  \frac{k}{n-1} - U_{(k)}$, we have  $\varepsilon<F_{n-1}\lp U_{(k)}\rp - U_{(k)}$. Also,  as $0 \leq U_{(k)} \leq 1$,  we have 

%  This is obtained by observing that if the inequality $F_{n-1}(x)-x>\varepsilon$ occurs for some $x$, then it must occur for $x=U_{(k)}$ for some $k$.
This proves \eqref{hgqewd}, and Step~1 is concluded by combining it with~\eqref{eq-wpjgrt}.

\textbf{Step 2:} We now proceed to show that if 
\begin{equation} \label{eq:eps}
	\varepsilon = \frac{\lfloor b \rfloor}{n-1}-\frac{1}{n},
\end{equation}
then, $E\subseteq E'(\varepsilon)$, where  event $E$ is defined in Proposition~\ref{prop:main}.

It is enough to show that for all $k \in \{\lfloor b \rfloor,\ldots,  n-1\}$, $\frac{k+1 - b}{n} \leq \frac{k - \lfloor b \rfloor}{n-1} + \frac{1}{n}$, which can be deduced from the following implications: 
$$ \frac{k+1 - b}{n} \leq \frac{k - \lfloor b \rfloor}{n-1} + \frac{1}{n} \Leftrightarrow \frac{k - b}{n} \leq \frac{k - \lfloor b \rfloor}{n-1} \Leftarrow  \frac{k - \lfloor b \rfloor}{n} \leq \frac{k - \lfloor b \rfloor}{n-1} \Leftrightarrow  \frac{1}{n} \leq \frac{1}{n-1}.$$

\textbf{Step 3:}
By Step 2, we have $ \P(E) \leq \P(E'(\varepsilon))$. Also,  observe that $ \lfloor\frac{b}{l} \rfloor\geq 1-\frac1n +\sqrt{\frac{(n-1)\log 2}{2}}$ implies $\varepsilon \geq \sqrt{\frac{\log 2}{2(n-1)}}$. Thus, combine it with Step 1 to obtain 
\begin{equation} \label{eq:e}
	\P(E)\leq \P(E'(\varepsilon)) \leq  \exp{\Big(-2(n-1) \Big(\frac{\lfloor b \rfloor}{n-1}-\frac{1}{n}\Big)^2\Big)}.
\end{equation}
Combine \eqref{eq:e} with Proposition~\ref{prop:main} to conclude the theorem. \qed

% \hossein{By Step 2, we have $ \P(E) \leq \P(E'(\varepsilon))$. Combine it with Step 1 to obtain
	% \begin{equation} \label{eq:e}
		%     \P(E) \leq  \exp{\lp-2(n-1) \lp\frac{\lfloor \frac{b}{l} \rfloor}{n-1}-\frac{1}{n}\rp^2\rp}
		% \end{equation}
	% provided that $\varepsilon = \frac{\lfloor \frac{b}{l} \rfloor}{n-1}-\frac{1}{n} \geq \sqrt{\frac{\log 2}{2(n-1)}}$, i.e., provided that
	% \begin{equation}
		%     % \lp \frac{\lfloor b\rfloor}{n-1}-\frac{1}{n}\rp^2\geq \frac{\log 2}{n-1}
		%     \lfloor \frac{b}{l} \rfloor \geq 1-\frac1n +\sqrt{\frac{\log 2}{2}}\sqrt{n-1}
		%      \label{eq-ghefq}
		% \end{equation}
	% Combine \eqref{eq:e} with Proposition~\ref{prop:main} to conclude the theorem. \qed}

% For the value of $\varepsilon $ obtained at Step 2, we have $E\subseteq E'(\varepsilon)$ and thus, combine \eqref{eq:s1DKW},  \eqref{eq:eps}, and Proposition~\ref{prop:main} to obtain

% \begin{equation}
	%     \P\lp B > b\rp \leq n\P(E) \leq n\P(E') \leq n \exp{\lp-2(n-1) \lp\frac{\lfloor \frac{b}{l} \rfloor}{n-1}-\frac{1}{n}\rp^2\rp}
	%     \label{eq-jhwdef}
	% \end{equation}
% provided that $\varepsilon \geq \sqrt{\frac{\log 2}{2(n-1)}}$, i.e., provided that
% \begin{equation}
	%     % \lp \frac{\lfloor b\rfloor}{n-1}-\frac{1}{n}\rp^2\geq \frac{\log 2}{n-1}
	%     \lfloor \frac{b}{l} \rfloor \geq 1-\frac1n +\sqrt{\frac{\log 2}{2}}\sqrt{n-1}
	%      \label{eq-ghefq}
	% \end{equation}
% which concludes the theorem. \qed
\end{proof}

Note that the bound of Theorem~\ref{thm:homog_dkw} is only less than one and is non-trivial when $  \lfloor\frac{b}{l} \rfloor \geq 1-\frac1n +\sqrt{\frac{(n-1)\log 2}{2}}$.

The following corollary provides a closed-form formulation for the minimum value for the aggregate burstiness with a violation probability of at most $\varepsilon$.  It is obtained by setting the right-hand side of \eqref{eq:epsdkw} in Theorem~\ref{thm:homog_dkw} to $\varepsilon$.
% \hossein{It is obtained by Theorem~\ref{thm:homog_dkw}, and by finding the smallest value of $b$ such that $\epdkw(n,l,b) \leq \varepsilon$.}

%  It is obtained by setting the right-hand side of \eqref{eq:epsdkw} in Theorem~\ref{thm:homog_dkw} to $\varepsilon$.

% \jylb{how about the condition on epsilon ? why does it disappear ?}
\begin{corollary}[Quasi-deterministic burstiness bound] \label{cor:homog_burst}
Assume model ${\bf (H)}$ with $n>1$. Consider some $0 < \varepsilon < 1$, and define 
\begin{equation}
	b(n, \ell, \varepsilon) \isdef \ell \Bigl\lceil 1 - \frac1n+\sqrt{\frac{(n-1)(\log\; n-\log\;\varepsilon)}{2}} \Bigr\rceil.
\end{equation}
Then, $b(n, \ell, \varepsilon)$ is a quasi-deterministic burstiness bound for the aggregate with the violation probability of at most $\varepsilon$, i.e., $\P\lp B > b\lp n, \ell, \varepsilon \rp \rp \leq \varepsilon$.
% \begin{equation}
	%      \P\lp \exists t \geq 0,~ B(t) > b_F^\epsilon\rp \leq \epsilon
	% \end{equation}
\end{corollary}
% \begin{proof}
%     Observe that $\frac{b(n,l, \varepsilon)}{l} = \lfloor\frac{b(n,l, \varepsilon)}{l} \rfloor\geq 1-\frac1n +\sqrt{\frac{(n-1)\log 2}{2}}$ as $\frac{b(n,l, \varepsilon)}{l} $ is an integer and $0 < \varepsilon < 1$. Then, apply Theorem~\ref{thm:homog_dkw} with $b(n,l, \varepsilon)$ and observe that $\epdkw\lp n ,b(n,l, \varepsilon), l  \rp  \leq \varepsilon$.\qed
% \end{proof}

% Corollary~\ref{cor:homog_burst} is obtained by setting the right-hand side of \eqref{eq:epsdkw} in Theorem~\ref{thm:homog_dkw} to $\varepsilon$, i.e., by finding the smallest $b$ such that $\epdkw(n,l,b) \leq \varepsilon$. 
Observe that $b(n, \ell, \varepsilon)$ grows in $\sqrt{n\log{n}}$ as opposed to the deterministic bound ($\ltot = n\ell$) that grows in linearly (see Fig.~\ref{fig:homog}b).
% , i.e., the ratio $\frac{b(n, l,\varepsilon)}{\ltot}$ is in the order of $\sqrt{\frac{\log{n}}{n}}$ (see Fig.~\ref{fig:num}b).
% We also numerically confirm that, given a non-zero violation probability, $b(n, l,\epsilon)$ is tight compared to the burstiness obtained by Theorem~\ref{thm:homog} (see Fig. ??). 
% We now present a refinement to Theorem~\ref{thm:homog_dkw} that  provides slightly  better bounds when the number of flows is small. 

Proposition~\ref{prop:main} introduces the event $E$ such that an upper bound of $\P(E)$ is used to derive an upper bound on the tail probability of the aggregate burstiness.   Theorem~\ref{thm:homog_dkw} is derived from the DKW upper bound of  $\P(E)$, which is tight when the number of flows $n$ is large. In Theorem~\ref{thm:homog}, we compute the exact value of $\P(E)$; thus, it provides a slightly better bound when the number of flows is small but at the expense of not having a closed-form expression.

% We now present a refinement to Theorem~\ref{thm:homog_dkw} that  provides slightly  better bounds when the number of flows is small.   Lemma~\ref{lem:main} introduces the event $E$ such that an upper bound of $\P(E)$ provides an upper bound on the tail probability of the aggregate burstiness.   Theorem~\ref{thm:homog_dkw} gives a closed-form expression that gives an upper bound to $\P(E)$ using the DKW inequality and  is tight when the number of flows $n$ is large. Theorem~\ref{thm:homog}  exactly computes $\P(E)$; thus, it provides a slightly better bound when the number of flows is small but at the expense of not having a closed-form expression.
% providing bounds less than or equal to that of Theorem~\ref{thm:homog_dkw}. However,  provides a slightly better bound when the number of flows is small, but at the expense of not having a closed-form expression.

\begin{theorem}[Refinement of Theorem~\ref{thm:homog_dkw} for small groups] \label{thm:homog}
Assume model ${\bf (H)}$ with $n>1$. For all $b \geq 0$. Then, a bound on the tail probability of the aggregate burstiness $B$ is
% of Section~\ref{sec:assump}, and also assume that flows have the same period and same packet size, i.e., for every flow $f$, $\tau_f = \tau$ and $l_f = l$ for some $l, \tau > 0$. Consider some $b$. Then, 
\begin{equation} \label{eq:epsThm1}
	\P\lp  B > b\rp \leq n(1 - p(n, \ell, b)) \isdef \ep(n, \ell, b),
\end{equation}
with 
\begin{align}
	\label{eq:pbThm1}
	p(n, \ell, b) & = (n-1)!\int_{y_{n-1} = u_{n-1}}^{1} \int_{y_{n-2} = u_{n-2}}^{y_{n-1}} \ldots \int_{y_{i} = u_{i}}^{y_{i+1}} \ldots \int_{y_{1} = u_{1}}^{y_{2}} 1   \,dy_1 \ldots \,dy_{n-1},
\end{align}
%	\label{eq:ukThm1}
	% p(b) & = (n-1)!\int_{y_{n-1} = u_{n-1}}^{1} \int_{y_{n-2} = u_{n-2}}^{y_{n-1}} \ldots \int_{y_{i} = u_{k}}^{y_{k+1}} \ldots \int_{y_{1} = u_{1}}^{y_{2}}    \,dy_1 \ldots \,dy_{n-1}
	and $u_k = \frac{[ (k+1) - b / \ell ]^+}{n}$, for all $k\in\Nats_{n-1}$ and $[x]^+=\max(0, x)$.
	% , \forall k \in \{1, \ldots, n-1 \}

%Also,  $\epdkw(n,l,b) \geq \ep(n,l,b)$ where $\epdkw(n,l,b)$ is the bound of Theorem~\ref{thm:homog_dkw}. 
\end{theorem}

Note that the computation of the bound of Theorem~\ref{thm:homog} requires computing $p(n,\ell, b)$ in \eqref{eq:pbThm1},  which is a series of polynomial integrations, and finding a general closed-form formula might be challenging. However, computing the bound can be done iteratively as in Algorithm~\ref{alg:thm1}: The integrals are computed from the inner sign to the outer (incorporation factor $i$ from the factorial in the $i$-th integral). Polynoms are computed at each step and variable $q_{j}^m$ represents the coefficient of degree $j$ of the $m$-th integral. Note that we always have $q_m^m=1$, so the monomial of degree $n-1$ cancels in~\eqref{eq:epsThm1}. 

All computations involve exact representations of the integrals (no numerical integration) and  use exact arithmetic with rational numbers; therefore, the results are exact with infinite precision.

% We iteratively compute the bound of Theorem~\ref{thm:homog} in Algorithm~\ref{alg:thm1}. 

% Note that $p(n,l,b)$ in \eqref{eq:pbThm1} is a series of polynomials integration, and finding a general close-form formula might be challenging. We iteratively compute the bound of Theorem~\ref{thm:homog} in Algorithm~\ref{alg:thm1}. \jylb{This last sentence is not well written. \eqref{eq:pbThm1} is the same as ``the bound of Theorem~\ref{thm:homog}" but this is not obvious. Also there is no link with the previous sentence. }

% However, we compute the bound of Theorem~\ref{thm:homog} iteratively in Algorithm~\ref{alg:thm1}. 

\begin{algorithm}[H]
%\scriptsize 
\SetKwInOut{IV}{Local Variables}
\SetKwInOut{Input}{Inputs}
\SetKwInOut{Output}{Output}
\Input{number of flows $n$, a burst $b$, and a packet size $\ell$.}
\Output{$\ep(n,\ell, b)$ such that $\P(B > b) \leq \ep(n,\ell, b)$.}
% , a bound on the tail probability of the aggregate burstiness, i.e., $\P(B > b) \leq \ep$.}
% \KwResult{$\ep$, a bound on the tail probability of the aggregate burstiness, i.e., $\P(B > b) \leq \ep$}
%\IV{integer $m$, $u_k$ for $k =\lfloor \frac{b}{l}\rfloor:n-1$, and collections $\lp q_1^m, \ldots, q_m^m \rp$ for $m=\lfloor\frac{b}{l} \rfloor-1:n-1$.}
% 	 \KwResult{$\lp \beta, d, \zcut,b\rp$: Collection of per-node strict service curve, per-node delay bounds,  bounds on propagated burstiness of flows at cuts, and  aggregate burst bound at the output of every node and every edge for all classes}

% $\forall k \in \{\lfloor b \rfloor - 1, \ldots, n-1 \}$, $u_k \gets$ as in \eqref{eq:ukThm1}\;
% $\bar{b} \gets \frac{b}{l}$\;
$m \gets \lfloor b / \ell \rfloor - 1$\;
$\lp q_0^m, q_1^m, \ldots, q_m^m \rp \gets \lp 0, 0, \ldots,0, 1 \rp$\;
\For{$m \leftarrow \lfloor b / \ell \rfloor$ {\bf to} $n-1$}{
$u_{m} \gets (m+1 - b / \ell ) / n$\;
$q_0^{m} \gets -\sum_{j=0}^{m-1} \frac{mq_j^{m-1}}{j+1}u_{m}^{j+1}$\; % \mand q_{m}^{m} \gets 1 $\;
\lFor{$i \leftarrow 1$ {\bf to} $m$}{
	$q_i^{m} \gets \frac{m q_{i-1}^{m-1}}{i} $
	% $\lp q_1^{m+1}, \ldots, q_{m+1}^{m+1} \rp \gets \Pi\lp \lp q_1^m, \ldots, q_m^m \rp , u_{m+1} \rp$\;
}
% $q_{m+1}^{m+1} \gets 1 $\;
}
%$\ep \gets  n\sum_{i=0}^{n-2} q_i^{n-1}$\;
\textbf{return} $n\sum_{i=0}^{n-2} q_i^{n-1}$
\caption{Computation of $\ep\lp n,\ell, b\rp$ from Theorem~\ref{thm:homog}}
\label{alg:thm1}
\end{algorithm}

\begin{proof}
%  We first show the following equation where recall that event $E$ is defined in Proposition~\ref{prop:main}:
% \begin{equation} \label{eq:order}
%     \P \lp E \rp = 1 - p(n,l,b)
% \end{equation}
Let $\Bar{E}$ be the complementary event of $E$ defined in Proposition~\ref{prop:main}. %Observe that, since $u_k = 0$ for all $k< \lfloor b \rfloor$, 
\begin{equation}
\Bar{E} = \bigcap_{k = 1}^{n-1} \Big\{ U_{(k)} \geq \frac{\lb k + 1 - b / \ell \rb^+}{n} \Big\}. %=  \Big\{ \forall k \in\Nats_{n-1},~U_{(k)} \geq \frac{\lb (k + 1) - b \rb^+}{n} \Big\}.
\end{equation}
Let $f_{U_{(1)},\ldots,U_{(n-1)}}$ be the density function of the joint distribution of $U_{(1)},\ldots, U_{(n-1)}$, given in \eqref{eq:orderUniform}. Then
% \cite[Equation 1.145]{gentle2009computational}
% \begin{equation}
%     f_{U_{(1)},...,U_{(n-1)}}\lp y_{1},...,y_{n-1} \rp  = (n-1)!\prod_{k=1}^{n-1}f_{U}(y_{k})\ind{y_{1}\leq y_2 \leq  \ldots \leq y_{n-1}}
% \end{equation}
% where $f_U$ is the probability density function of a standard uniform; as $y_{k} \in [0,1]$, $f_{U}(y_{k}) = 1$ for $k = 1:n-1$
% % $k \in \{1,\ldots,n-1\}$
% . Hence,
% \begin{equation}
%     f_{U_{(1)},...,U_{(n-1)}}\lp y_{1},...,y_{n-1} \rp  = (n-1)!\ind{y_{1}\leq y_2 \leq  \ldots \leq y_{n-1}}
% \end{equation}
%$\P \lp \Bar{E} \rp $ is equal to the probability that  $k = 1:n-1, U_{(k)} \geq \frac{(k + 1)l - b}{\ltot} = u_k$, hence 
% $\P \lp \Bar{E} \rp $
% $ \P \lp \bigcap_{k = 1}^{n-1} \lc U_{(k)} \geq \frac{\lb (k + 1)l - b \rb^+}{\ltot} \rc \rp = $
{\footnotesize 
\begin{align}
	\P \lp \Bar{E} \rp &=  \int_{y_{n-1} = u_{n-1}}^{1} \ldots \int_{y_{i} = u_{i}}^{1} \ldots \int_{y_{1} = u_{1}}^{1} f_{U_{(1)},\ldots,U_{(n-1)}}\lp y_{1},\ldots,y_{n-1} \rp    \,dy_1 \ldots \,dy_{n-1} \\
	&=\int_{y_{n-1} = u_{n-1}}^{1} \ldots \int_{y_{i} = u_{i}}^{1} \ldots \int_{y_{1} = u_{1}}^{1} (n-1)!\ind{0 \leq y_{1}\leq y_2 \leq  \ldots \leq y_{n-1} \leq 1}   \,dy_1 \ldots \,dy_{n-1} \\
	&=(n-1)!\int_{y_{n-1} = u_{n-1}}^{1}  \ldots \int_{y_{i} = u_{i}}^{y_{i+1}} \ldots \int_{y_{1} = u_{1}}^{y_{2}} 1   \,dy_1 \ldots \,dy_{n-1} = p(n,\ell, b).
	% \\
	% &= p(n,l,b)
\end{align}
}
Combine it with  $\P \lp \Bar{E} \rp = 1 - \P \lp E \rp$ and Proposition~\ref{prop:main} to conclude the theorem. \qed
% \textbf{Final Step (Union Bound)}: We proceed to prove the theorem. By Step 1, 
% % It follows that
% \begin{equation} \label{eq:s41}
%    \lc B > b \rc  = \bigcup_{i=1}^n \lc B_i > b \rc \Rightarrow \P \lp B > b \rp  = \sum_{i=1}^n \P \lp B_i > b \rp
% \end{equation}
% % \begin{equation} \label{eq:s41}
% %    \P \lp B > b \rp  = \bigcup_{i=1}^n \P \lp B_i > b \rp 
% % \end{equation}
% Combining Step 2 and Step 3, we have 
% \begin{equation}\label{eq:s42}
%    % \forall i \in \{1,\ldots,n\}, 
%    \P \lp B_i > b \rp  = 1 -  p(n,l,b), \;\;\; i = 1:n
% \end{equation}
% Combining \eqref{eq:s41} and \eqref{eq:s42} we have 
% \begin{equation} \label{eq:s43}
%    \P \lp B > b \rp  = n\P \lp B_i > b \rp =  n\lp 1 -  p(n,l,b)\rp 
% \end{equation}
% which concludes  the theorem. \qed
\end{proof}

Note that since Theorem~\ref{thm:homog} computes the exact probability of event $E$, we have $\epdkw(n,\ell, b) \geq \ep(n, \ell, b)$.

\section{Heterogeneous Case} \label{sec:heterogen}

In this section, we consider the case where flows have different periods and packet sizes. We present burstiness bounds in two different settings: First, when flows can be grouped into homogeneous flows; second, when all packets have the same period but with different packet sizes.

%two theorems; the former obtains a bound by grouping flows into  homogeneous sets and combining the bounds obtained for each set, using a convolution bounding technique, which we present in Proposition~\ref{prop:conv}. The latter, for a specific case, when flows have the same period and different packet sizes, provides an alternative bound that may be better when the number of flows per packet size is small.

% \hossein{
	% \begin{enumerate}
		%     \item Theorem 3: Same Period but Different Packet sizes
		%     \item Theorem 4: Different Period and Packet sizes, Convolution Bound (Heterogeneous flows)
		% \end{enumerate}}

Let us first focus on the model where flows are grouped according to their characteristics:
\begin{itemize}
	\item[{\bf (G)}] There exists a partition $I_1, \ldots, I_g$ of $\Nats_n$ such that  $I_i$ is a group of $n_i$ flows satisfying model ${\bf (H)}$ with packet size $\ell_i$ and period $\tau_i$. All phases are mutually independent. 
\end{itemize}

\begin{proposition}[Convolution Bound]\label{prop:conv}
	Let $X_1,X_2,\ldots,X_g$ be $g \geq 1$ mutually independent rv on $\Nats$. Assume that for all $i\in\Nats_n$, $\Psi_i$ is wide-sense increasing and is a lower bound on the CDF of $X_i$, namely, $\forall b\in\Nats$, $\P(X_i \leq b) \geq \Psi_i(b)$. Define $\psi_i$ by $\psi_i(0) = \Psi_i(0)$ and $\psi_i(b)=\Psi_i(b)-\Psi_i(b-1)$ for $b\in \Nats\setminus\{0\}$. 
	
	Then, a lower bound on the CDF of $\sum_{i=1}^gX_i$ is given by: $\forall b\in\Nats$, 
	\begin{equation}\label{eq:convBound}
		\P \Big(\sum_{i=1}^gX_i \leq b \Big) \geq \lp\psi_1*\psi_2*\cdots*\psi_{g-1}*\Psi_{g} \rp(b), %\Psi_{1,2,\ldots,g}(b),
	\end{equation}
	%with  
	%\begin{equation}\label{eq:psi}
%		\Psi_{1,2,\ldots,g}(b) \isdef \lp\psi_1*\psi_2*\ldots*\psi_{g-1}*\Psi_{g} \rp(b), \;\;\; b\in \Nats
%	\end{equation}
	where, the symbol $*$ denotes the discrete convolution, defined for arbitrary functions $f_1,f_2:\Nats \to \mathbb{R}$ by
  \vspace{-0.01cm}
	\begin{equation}\label{eq:discConv}
		\forall b\in\Nats, \quad (f_1* f_2)(b) = \sum_{j=0}^{b} f_1(j) f_2(b-j).
	\end{equation}
\end{proposition}
\begin{proof}
	% We first  show the result when $g=2$, then extend it to $g > 2$. 
	We prove it by induction on $g$.
	
	\noindent \textbf{Base Case $g=1$:} There is nothing to prove: for all $b\in\Nats$, $\P(X_1 \leq b) \geq \Psi_1(b)$. 
	
	\noindent \textbf{Induction Case:}
	We now assume that Equation~\eqref{eq:convBound} holds for $g$ variables, and we show that it also holds for $g+1$ variables. 
	
	We can apply Equation~\eqref{eq:convBound} to  variables $X_2, X_3, \ldots, X_{g+1}$, and let us denote $Y = X_2 + \cdots + X_{g+1}$ and  $\Psi=\psi_2*\cdots*\psi_g * \Psi_{g+1}$. 
	We need to show that for all $b\in\Nats$, 
	% $\P( X_1 + X_2 \leq b) \geq \Psi_{1,2}(b)$, for $b\in \Nats$, with
	\begin{equation} \label{eq:conv}
		\P( X_1 + Y \leq b) \geq (\psi_1 * \Psi)(b).
		% =  (\psi_2 * \Psi_1)(b)
	\end{equation}
	
	% we show that $\P( X_1 + X_2 \leq b) \geq \Psi_{1,2}(b) = (\psi_1 * \Psi_2)(b)$.
	% Consider two groups, say $1$ and $2$; we show that
	% $\P( X_1 + X_2 \leq b) \geq \Psi_{1,2}(b)$, for $b\in \Nats$, with
	% \begin{equation} \label{eq:conv}
		%     \Psi_{1,2}(b)= (\psi_1 * \Psi_2)(b) =  (\psi_2 * \Psi_1)(b)
		% \end{equation}
	
	Let $F(b)=\P(Y\leq b)$ and observe that $\P(Y = 0) = F(0)$ and $\P(Y=b)=F(b)-F(b-1)$ for $b\in\Nats\setminus\{0\}$. Then, since $X_1$ and $Y$ are independent, 
	{%\footnotesize
		\begin{align}
			\P(X_1+ Y\leq b)  &=\sum_{j=0}^{b}\P(X_1+j\leq b | Y=j)\P(Y= j) = \sum_{j=0}^{b}\P(X_1+j\leq b )\P(Y= j)\\
			% &=&    \sum_{j=0}^{b}\P(X_1+j\leq b )\P(X_2= j)\\
			%&= \sum_{j=0}^{b}F_1(b-j)\P(X_2= j) 
			&\geq \sum_{j=0}^{b}\Psi_1(b - j)\P(Y= j) 
			\\
			% &\geq&\sum_{j=0}^{b}\Psi_1(b - j)\P(X_2= j) \\
			&\geq\Psi_1(b)F(0) + \sum_{j=1}^{b}\Psi_1(b - j)(F(j)-F(j-1)).
			\label{eq:last}
	\end{align}}
	%
	% \begin{eqnarray}
		%     &=&\sum_{j=0}^{b}\P(X_1+j\leq b | X_2=j)\P(X_2= j) = \sum_{j=0}^{b}\P(X_1+j\leq b )\P(X_2= j)\\
		% % &=&    \sum_{j=0}^{b}\P(X_1+j\leq b )\P(X_2= j)\\
		% &=& \sum_{j=0}^{b}F_1(b-j)\P(X_2= j) \geq \sum_{j=0}^{b}\Psi_1(b - j)\P(X_2= j) 
		% \\
		%     % &\geq&\sum_{j=0}^{b}\Psi_1(b - j)\P(X_2= j) \\
		%     &\geq&\Psi_1(b)F_2(0) + \sum_{j=1}^{b}\Psi_1(b - j)(F_2(j)-F_2(j-1)) 
		%     \label{eq:last}
		% \end{eqnarray}
	% \begin{eqnarray}
		%     \P(X_1+X_2\leq b)&=&\sum_{j=0}^{b}\P(X_1+j\leq b | X_2=j)\P(X_2= j)\\
		% &=&    \sum_{j=0}^{b}\P(X_1+j\leq b )\P(X_2= j)\\
		% &=& \sum_{j=0}^{b}F_1(b-j)\P(X_2= j)
		% \\
		%     &\geq&\sum_{j=0}^{b}\Psi_1(b - j)\P(X_2= j) \\
		%     &\geq&\Psi_1(b)F_2(0) + \sum_{j=1}^{b}\Psi_1(b - j)(F_2(j)-F_2(j-1)) 
		%     \label{eq:last}
		% \end{eqnarray}
	% \begin{align}
		% \P(X_1+X_2\leq b)&=\sum_{j=0}^{b}\P(X_1+j\leq b | X_2=j)\P(X_2= j) \nonumber \
		% &=\sum_{j=0}^{b}\P(X_1+j\leq b)\P(X_2= j) \nonumber \
		% &=\sum_{j=0}^{b}F_1(b-j)\P(X_2=j) \nonumber \
		% &\geq\sum_{j=0}^{b}\Psi_1(b-j)\P(X_2=j) \nonumber \
		% &\geq\Psi_1(b)F_2(0) + \sum_{j=1}^{b}\Psi_1(b-j)(F_2(j)-F_2(j-1)) \label{eq:last}
		% \end{align}
	% \begin{align}
		%         &= \sum_{j=0}^{b}\P(X_1+j\leq b | X_2=j)\P(X_2= j) =    \sum_{j=0}^{b}\P(X_1+j\leq b )\P(X_2= j)\\
		%         &=   \sum_{j=0}^{b}F_1(b-j)\P(X_2= j) \geq \sum_{j=0}^{b}\Psi_1(b - j)\P(X_2= j)\\
		% % &=   \sum_{j=0}^{b}F_1(b-j)\P(X_2= j)
		% % \\
		%     & \geq \Psi_1(b)F_2(0) + \sum_{j=1}^{b}\Psi_1(b - j)(F_2(j)-F_2(j-1))
		%     % \\
		%     % &\geq \Psi_1(b)F_2(0) + \sum_{j=1}^{b}\Psi_1(b - j)(F_2(j)-F_2(j-1)) 
		%     \label{eq:last}
		% \end{align}
	We now use Abel's summation by parts in \eqref{eq:last} and obtain
	\begin{align}
		% \P(X_1+X_2\leq b) 
		\P(X_1+Y\leq b) &\geq\Psi_1(b)F(0) + \sum_{j=1}^{b}\Psi_1(b - j)F(j)- \sum_{j=1}^{b}\Psi_1(b - j)F(j-1)
		\\
		&=
		\Psi_1(b)F(0) + \sum_{j=1}^{b}\Psi_1(b - j)F(j)- \sum_{j=0}^{b-1}\Psi_1(b - j - 1)F(j)
		\\
		&=
		\sum_{j=0}^{b}\Psi_1(b - j)F(j)- \sum_{j=0}^{b-1}\Psi_1(b - j - 1)F(j)
		\\
		&=
		\Psi_1(0)F(b) + \sum_{j=0}^{b-1}(\Psi_1(b - j) - \Psi_1(b - j- 1))F(j)\\
		&=
		\psi_1(0)F(b) + \sum_{j=0}^{b-1}\psi_1(b - j)F(j) = \sum_{j=0}^{b}\psi_1(b - j)F(j)
		\\
		&
		 \geq \sum_{j=0}^{b}\psi_1(b - j)\Psi(j)=  (\psi_1 * \Psi)(b).
	\end{align}
	We can conclude by using the associativity of the discrete convolution: $\psi_1 * \Psi = \psi_1 * \cdots * \psi_g * \Psi_{g+1}$.
	% We now iteratively apply \eqref{eq:conv} and obtain $ \P\lp \sum_{i=1=:g}  X_i \leq b \rp \geq \lp\psi_1*\psi_2*\ldots*\psi_{g-1}*\Psi_{g} \rp(b) =  \Psi_{1,2,\ldots,g}(b)$ for $b\in \Nats$, 
	% % \begin{equation} \label{eq:hs3}
		% %     \Psi_{1,2,\ldots,g}(b) = \lp\psi_1*\psi_2*\ldots*\psi_{g-1}*\Psi_{g} \rp(b), \;\;\; b>0,b\in \Nats
		% % \end{equation}
	% which concludes the proposition.
	\qed
\end{proof}

\noindent {\bf Remarks.} 1. Note that  $(\psi_1 * \Psi_2)(b) = \sum_{i+j \leq b} \psi_1(i) + \psi_2(j) = (\psi_2 * \Psi_1)(b)$, so the convolution bound is independent of the order of $X_1, \ldots, X_g$. 

%Note that $\psi_1 * \Psi_2 = \Psi_1 * \psi_2$ as $ (\psi_1 * \Psi_2)(b) =\sum_{j=0}^{b}\Psi_1(b - j)\Psi_2(j) - \sum_{j=1}^{b}\Psi_1(b - j)\Psi_2(j) = \sum_{i,j s.t. i + j = b, i,j \geq 0 }\Psi_1(i)\Psi_2(j)  +  \sum_{i,j s.t. i + j = b - 1, i,j \geq 0 }\Psi_1(i)\Psi_2(j)$ which is symmetric in indices 1,2. Thus, the convolution bound is independent of the labeling of the variables, i.e.,  $ \psi_1*\psi_2*\ldots*\psi_{g-1}*\Psi_{g} = \psi_1*\psi_2*\ldots*\psi_{g-2}*\psi_{g}*\Psi_{g-1} = \ldots =   \psi_2*\ldots*\psi_{g-1}*\psi_{g}*\Psi_{1} $. 

% in \eqref{eq:convBound}, the index of $\Psi_g$  can be swapped by the index of any $\psi_i$ for $i = 1:g-1$ without affecting the bound, i.e., $ \psi_1*\psi_2*\ldots*\psi_{g-1}*\Psi_{g} = \psi_1*\psi_2*\ldots*\psi_{g-2}*\psi_{g}*\Psi_{g-1} = \ldots =   \psi_2*\ldots*\psi_{g-1}*\psi_{g}*\Psi_{1} $. 

\noindent 2.  An alternative  to Proposition~\ref{prop:conv} is to use then union bound rather than the convolution bound: for all $(b_1, \ldots, b_g) \in \Nats^g$ such that $\sum_{i=1}^g b_i = b$, we have $\big\{ \sum_{i=1}^g X_i  > b \big\} \subseteq \bigcup_{i=1}^g\lc X_i > b_i \rc$, so $\P(X >  b) \leq \sum_{i=1}^g \P(X_i > b_i) \leq \sum_{i=1}^g (1-\Psi_i(b_i))$. We can choose $(b_i)_{i=1}^g$ so as to minimize this latter term, and take the complement to obtain 
\begin{equation}
	\P( \sum_{i=1}^{g}X_i  \leq  b) \geq 1 - \min_{b_1 + \cdots + b_g = b} \sum_{i=1}^g (1-\Psi_i(b_i)).
\end{equation}
This bound is also valid when rvs $X_i$ are not independent, but it can be shown that the convolution bound always dominates the union bound. In our numerical evaluations, we find that the convolution bound provides significantly better results than the union bound. 

\begin{theorem}[Flows with different periods and different packet-sizes] \label{thm:hetero}
	Assume model ${\bf (G)}$.  Let $\varepsilon_i$ be a wide-sense decreasing function that bounds the tail probability of  aggregate burstiness $B_i$ of each group $i \in\Nats_g$: for all $b\in\Nats$,  $ \P\lp B_i > b\rp \leq \varepsilon_i(b)$  for all $b\in \Nats$.
	% $\varepsilon_i(b)$  which is a bound on the tail probability of $B_i$, the aggregate burstiness of group $i$, using Theorem~\ref{thm:homog_dkw}, namely,  $ \P\lp B_i > b\rp \leq \varepsilon_i(b)$  for all $b\in \Nats$.
	Define $\Psi_i(b) = 1 - \varepsilon_i(b)$ for $b\in \Nats$ and define $\psi_i$ by $\psi_i(0) = \Psi_i(0)$ and $\psi_i(b)=\varepsilon_i(b-1)-\varepsilon_i(b)$ for $b\in \Nats \setminus\{0\}$.
	
	Then, a bound on the tail probability of the aggregate burstiness of all flows $B$ is given by $\forall b \in \Nats_{\ltot}$,
	\begin{equation} \label{eq:eps_hetero}
		\P\lp B > b\rp \leq 1 - \lp\psi_1*\psi_2*\cdots*\psi_{g-1}*\Psi_{g} \rp(b),
	\end{equation}
\end{theorem}

\begin{proof} 
% We first show the following equation
% \begin{equation} \label{eq:hs1}
	%     \P(  B \leq b) \geq \P( \sum_{i=1:g} \lceil B_i\rceil \leq b ), \;\;\; b\in \Nats
	% \end{equation} 
% and for $b \in \Nats$,
% \begin{equation}
	%     \P\lp \sum_{i=1}^g\sup_{t \geq 0}B_i(t) > b\rp \leq \P\lp \sum_{i=1}^g \lceil \sup_{t \geq 0}B_i(t) \rceil > b\rp 
	% \end{equation}

For all group $i\in\Nats_g$, let $A^i[s, t)$ be the aggregate of flows of group $i$ during the interval $[s,t)$, $r^i$, its aggregate arrival rate, and $B_i$ its aggregate burstiness. Observe that for all $s\leq t$, $A(s,t] = \sum_{i=1}^g A^i[s,t)$ and $\rtot = \sum_{i=1}^g r^i$. We then obtain 
\begin{align}
B & = \sup_{0\leq s \leq t} \{ A(s,t] - \rtot(t - s) \} =  \sup_{0\leq s \leq t} \Big\{  \sum_{i=1}^g \big( A_i(s,t] - r_i(t - s) \big)  \Big\} \\ 
& \leq  \sum_{i=1}^g \sup_{0\leq s \leq t} \lc   A_i(s,t] - r_i(t - s) \rc = \sum_{i=1}^g B_i\leq  \sum_{i=1}^g\lceil B_i\rceil.
	%     \Bar{B}(t) = \sup_{s \leq t} \{ A(s,t] - \rtot(t - s) \} &=  \sup_{s \leq t} \lc  \sum_{i=1}^g \lp A_i(s,t] - r_i(t - s) \rp \rc\\
	%     &\leq  \sum_{i=1}^g \underbrace{\sup_{s \leq t} \lc   A_i(s,t] - r_i(t - s) \rc}_{\Bar{B}_i(t)}
\end{align}
%     \begin{align}
	%     \Bar{B}(t) = \sup_{s \leq t} \{ A(s,t] - \rtot(t - s) \} &=  \sup_{s \leq t} \lc  \sum_{i=1}^g \lp A_i(s,t] - r_i(t - s) \rp \rc\\
	%     &\leq  \sum_{i=1}^g \underbrace{\sup_{s \leq t} \lc   A_i(s,t] - r_i(t - s) \rc}_{\Bar{B}_i(t)}
	% \end{align}
%     B(t) = \sup_{s \leq t} \{ A(s,t] - \rtot(t - s) \} =  \sup_{s \leq t} \lc  \sum_{i=1}^g \lp A_i(s,t] - r_i(t - s) \rp \rc \leq \sum_{i=1}^g
% \end{equation}
Hence, it follows that $ \P(  B \leq b) \geq \P( \sum_{i=1}^g \lceil B_i\rceil \leq b ), \;\;\; b\in \Nats$ .
% shows \eqref{eq:hs1}.

We now apply Proposition~\ref{prop:conv} with  $X_i = \lceil B_i\rceil$ and $\Psi_i$ as defined in the theorem: it suffices to observe that $(\lceil B_i\rceil)_{i\in\Nats_g}$ are mutually independent rv on $\Nats$; as $\varepsilon_i$ is wide-sense decreasing, $\Psi_i$ is  wide-sense increasing; Hence, by Proposition~\ref{prop:conv}, we obtain that for all $b\in\Nats$, $\P( \sum_{i=1}^{g} \lceil B_i\rceil \leq b) \geq (\psi_1*\psi_2*\ldots*\psi_{g-1}*\Psi_{g})(b)$, which concludes the proof.
% \begin{equation} \label{eq:hs2}
%         \P( \sum_{i=1:g} \lceil B_i\rceil \leq b) \geq \Psi_{1,2,\ldots,g}(b), \;\;\; b\in \Nats
% \end{equation} 
%Combine it with the previous paragraph to conclude the theorem. 
\qed
\end{proof}

We now turn to our second heterogeneous model: when all flows have the same period but different packet sizes. 
\begin{itemize}
	\item[${\bf (P)}$] There exists $\tau>0$ such that $\forall f\in\Nats_n$,~$\tau_f = \tau$; $\ell_1 \geq \ell_2 \geq \cdots \geq \ell_n>0$ and $(\phi_f)_{f\in\Nats_n}$ is a family of iid  uniform rv on $[0, \tau)$.
\end{itemize}

\begin{theorem}[Flows with the same period but different packet sizes] \label{thm:semi_homog}
Assume model ${\bf (P)}$. For all $0 \leq b < \ltot$,  set $\eta \isdef \min \Big\{  \frac{k}{n-1}-\frac{\sum_{j=1}^{k+1}\ell_j}{\ltot},~k\in\Nats_{n-1}, \sum_{j=1}^{k+1}\ell_j   > b\Big\}$. Then 

\begin{enumerate}
% $\P\lp \exists t \geq 0,~ B(t) > b\rp \leq n(1 - \bar{p}(n,l,b)) = \ep(n,l,b)$ with $\bar{p}(n,l,b)$ is computed as \eqref{eq:pbThm1} when replacing $u_1, u_2, \ldots,u_{n-1}$ by $\bar{u}_1, \bar{u}_2, \ldots,\bar{u}_{n-1}$ where $\bar{u}_k = \frac{[ \sum_{j=1}^{k+1}l_{(j)} - b]^+}{\ltot}, \forall k \in \{1, \ldots, n-1 \}$.

\item A bound on the tail probability of the aggregate burstiness of all flows $B$ is
% \begin{equation}\label{eq:epsdkw}
	%      \P\lp \exists t \geq 0,~ B(t) > b\rp \leq n \exp \lp -2(n-1) \lp \frac{\bar{k}}{n-1}-\frac{\Lcum(\bar{k}+1)-b}{\ltot} \rp^2 \rp  = \epdkw(n,l,b)
	% \end{equation}
\begin{equation}
	\P( B > b) \leq n\;\exp {\Big( -2(n-1) \big( \eta +\frac{b}{\ltot} \big)^2 \Big)}.
	% = \epdkw(n,l,b)
\end{equation}
% where $\bar{k} = \argmin_{k\in\lc 1,2,\ldots,n-1\rc, \sum_{j=1}^{k+1}l_{(j)} > b}{ \lc  \frac{k}{n-1}-\frac{\sum_{j=1}^{k+1}l_{(j)}}{\ltot} \rc}$.

\item For all $\varepsilon\in (0, 1)$, for all $n \geq 2$, the violation probability of at most $\varepsilon$, i.e., $\P\lp B > b(n,\ell_1,\ldots ,\ell_n, \varepsilon)\rp \leq \varepsilon$ with
\begin{equation}
	b(n,\ell_1,\ldots ,\ell_n, \varepsilon) \isdef \ltot \Big\lceil \sqrt{\frac{\log\; n - \log\; \varepsilon}{2(n-1)}}  - \eta  \Big\rceil .
\end{equation}

\item A bound on the tail probability of the aggregate burstiness of all groups $B$ is given by
%\begin{equation}
$	\P\lp B > b\rp \leq n(1 - \bar{p}(n,\ell_1,\ldots ,\ell_n,b)) $, where $\bar{p}(n,\ell_1,\ldots ,\ell_n,b)$ is computed as in Equation~\eqref{eq:pbThm1}, where for all $k\in\Nats_{n-1}$, $u_k = \frac{[ \sum_{j=1}^{k+1}\ell_j - b]^+}{\ltot}$.
\end{enumerate}
\end{theorem}

When all flows have the same packet-sizes $\ell$, this is model ${\bf (H)}$ and the bounds provided are exactly the same as in Section~\ref{sec:homogen}. Algorithm~\ref{alg:thm1} can also be used to compute the bound of item 3 if a) line $1$ is replaced by  $m \gets \max\{ k \geq 0~|~\sum_{j=1}^{k+1}\ell_j \leq b\}$ and  b) the values of $u_m$ are adapted in line 4.

\begin{proof}
	The proof is done by adapting Proposition~\ref{prop:main}. Then the proofs of each item follow exactly the steps of Theorems~\ref{thm:homog_dkw}, Corollary~\ref{cor:homog_burst} and Theorem~\ref{thm:homog}. 
%The proof follows similar steps as in the proof of Theorem~\ref{thm:homog_dkw},~Corollary~\ref{cor:homog_burst} and Theorem~\ref{thm:homog}. 
The key difference in Proposition~\ref{prop:main} is the computation of $H_{i,j}$: 
$H_{i,j} \leq \sum_{k=1}^{j - i +1}\ell_k  - \ltot(T_j - T_i)$: we bound this value as if the packets arrived in this arrival where the $j-i+1$ longest ones. 
\qed
\end{proof}

% \textbf{Step 3:} We now apply Step 2, \eqref{eq:conv}, with $\Psi_i$ and $\psi_i$ as defined in the theorem.

% The bound can be applied iteratively to the sum of $n$ independent random variables. 
% \input{proofs}
\section{Numerical Evaluation} \label{sec:num}
\begin{figure}[t]
     \centering
     \begin{subfigure}[b]{0.45\textwidth}
         \centering
         \includegraphics[width=\textwidth]{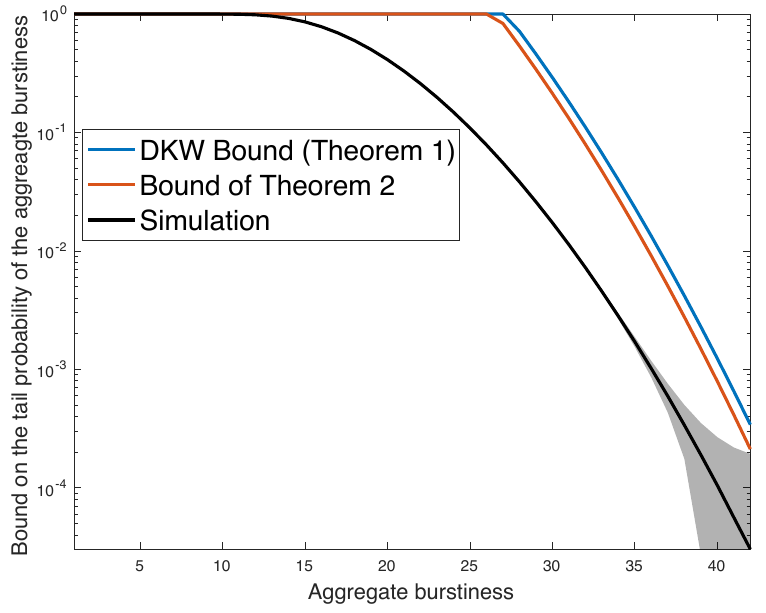}
         \caption{}
         \label{fig:homog_bound}
     \end{subfigure}
     \hfill
     \begin{subfigure}[b]{0.45\textwidth}
         \centering
         \includegraphics[width=\textwidth]{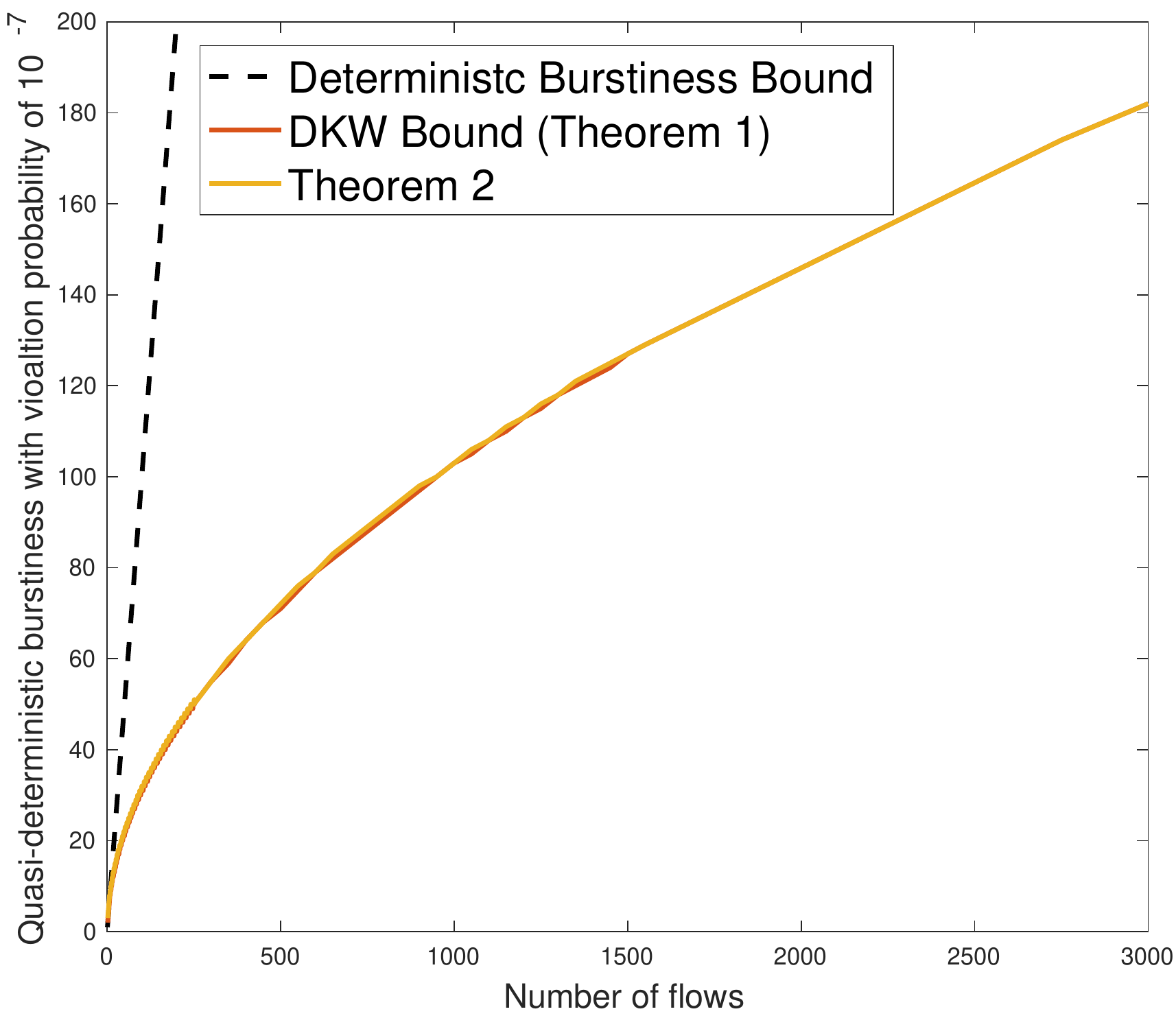}
          \caption{}
         \label{fig:homog_burst}
     \end{subfigure}
        \caption{\sffamily \small (a): Bound on the tail probability of the aggregate burstiness obtained by Theorems~\ref{thm:homog_dkw}, \ref{thm:homog}, and simulations. (b): The obtained quasi-deterministic burstiness with violation probability of $10^{-7}$ by Theorem~\ref{thm:homog_dkw} and Theorem~\ref{thm:homog}, as the number of flows grows; the deterministic bound (dashed plot) grows linearly with the number of flows.}
        \label{fig:homog}
\end{figure}

\begin{figure}[t]
     \centering
     \begin{subfigure}[b]{0.45\textwidth}
         \centering
         \includegraphics[width=\textwidth]{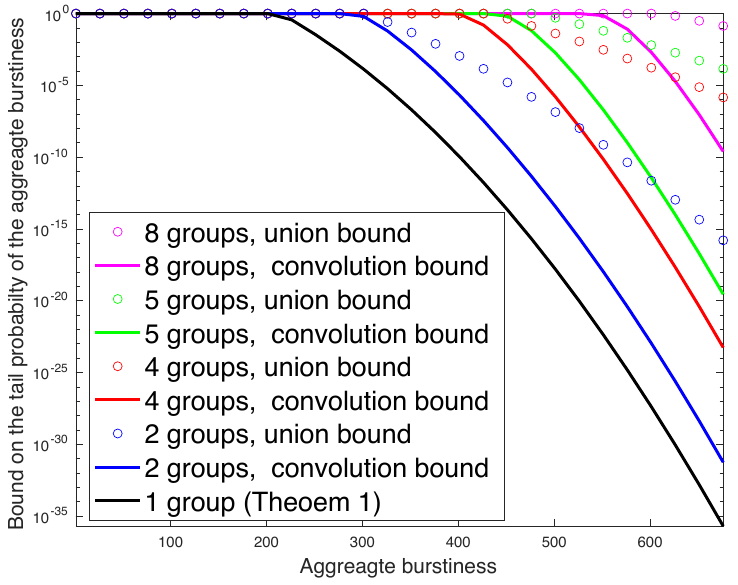}
          \caption{}
         \label{fig:heter_conv_union}
     \end{subfigure}
     \hfill
     \begin{subfigure}[b]{0.45\textwidth}
         \centering
         \includegraphics[width=\textwidth]{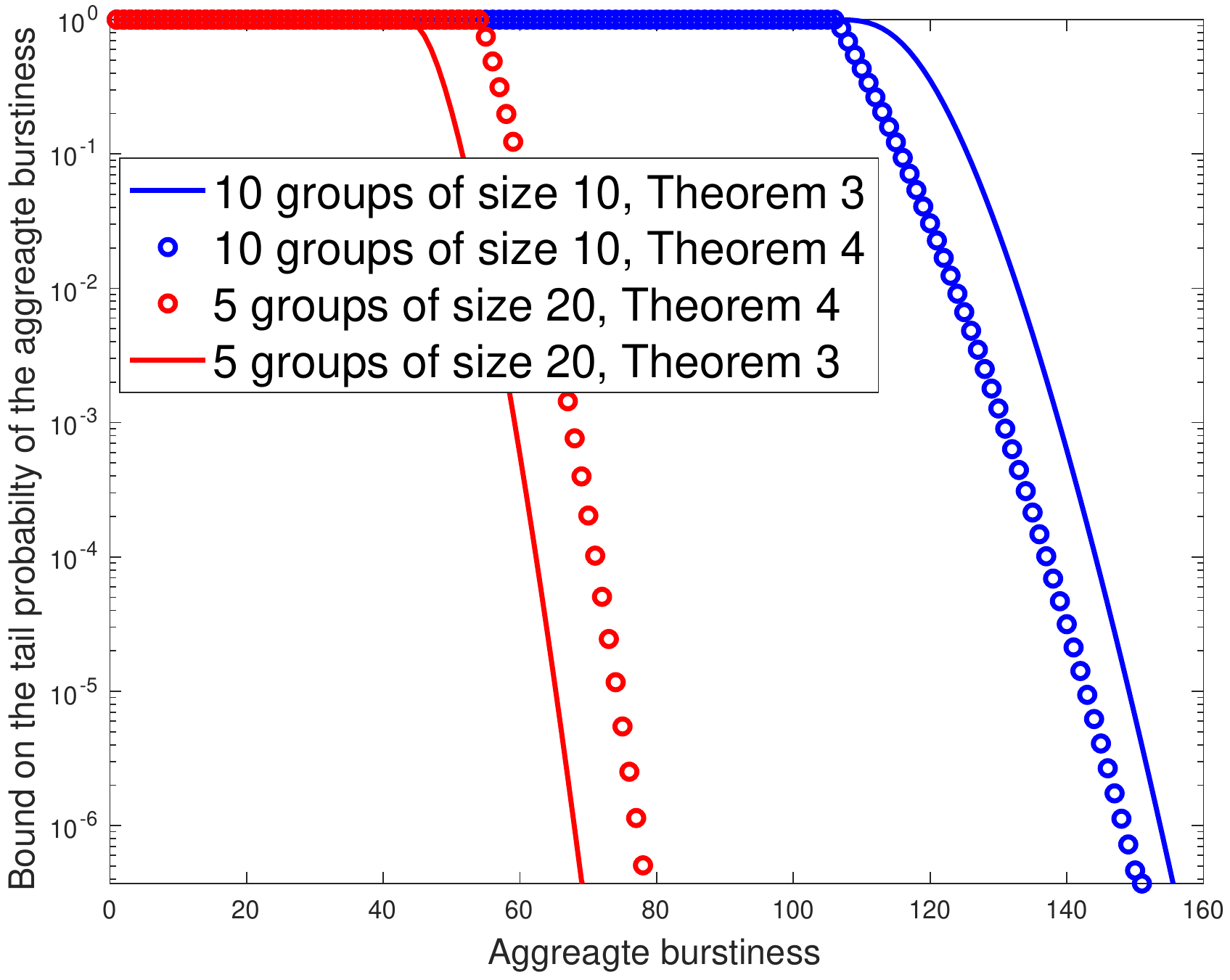}
          \caption{}
         \label{fig:heter_conv_packet}
     \end{subfigure}
        \caption{\sffamily \small (a): Comparison of the convolution bound of Theorem~\ref{thm:hetero} to the union bound when combining bound obtained for homogeneous sets of flows. (b): Slight improvement of Theorem~\ref{thm:semi_homog} compared to Theorem~\ref{thm:hetero} when the number of flows per packet-size is small.}
        \label{fig:heter}
\end{figure}

In this section, we numerically illustrate our bounds in Fig.~\ref{fig:homog} and Fig.~\ref{fig:heter}.

\subsection{Homogeneous Case}
In Fig.~\ref{fig:homog_bound}, we consider $250$ flows with the same packet size  (with respect to a unit, is assumed to be $1$) and the same period. We then compute bounds on the tail probability of their aggregate burstiness using Theorems~\ref{thm:homog_dkw} and \ref{thm:homog}.  We also compute the bound using simulations: For each flow, we independently pick a phase uniformly at random, and we then compute the aggregate burstiness as in \eqref{eq:agg_burst}; we repeat this $10^8$ times. We then compute bounds on the tail probability of their aggregate burstiness and its $99 \%$ Kolmogorov–Smirnov confidence band. The  bound of Theorem~\ref{thm:homog} is slightly better than that of Theorem~\ref{thm:homog_dkw}. Also,  compared to simulations, our bounds are fairly tight. 

In Fig.~\ref{fig:homog_burst}, we consider  $n\in\{2, \ldots, 3000\}$ flows with the packet size 1 and same period. We then compute a quasi-deterministic burstiness bound with violation probability of $10^{-7}$ once using Corollary~\ref{cor:homog_burst} and once using Theorem~\ref{thm:homog}; they are almost equal and  as $n$ grows are exactly equal, as Theorem~\ref{thm:homog_dkw} is as tight as Theorem~\ref{thm:homog} for large $n$. Also, our quasi-deterministic burstiness bound is considerably less than the deterministic one (i.e., $n$) and grows in $\sqrt{n\;\log\; n}$.

\subsection{Heterogeneous Case}
To assess the efficiency of the bound in the heterogeneous case, we consider in Fig.~\ref{fig:heter_conv_union} 10000 homogeneous flows with period and packet length 1, and divide them into $g$ groups of $10000/g$ flows, for $g\in \{1,2,4,5,8\}$. We compute a bound for each group by Theorem~\ref{thm:homog_dkw}, and combine them once with the convolution bound of Theorem~\ref{thm:hetero} and once by the union bound (as explained after Proposition~\ref{prop:conv}). Our convolution bound is significantly better than the union bound, and the differences increases fast with the number of sets.
%increases, the difference becomes more considerable. 

In Fig.~\ref{fig:heter_conv_packet}, we consider  $10$ (resp. $5$) homogeneous groups of $10$ (resp. $20$) flows, flows of each set $g\in\Nats_{10}$ (resp. $g\in\Nats_{5}$), have a packet-size equal to $g$, and all flows have the same period.
% We consider a $100$ flows with the same period but different packet sizes. Specifically, we assume $10$ homogeneous sets of $10$ flows, and flows of each set $g=1:10$, have a packet-size equal to $g$.
We then compute the bound on the tail probability of the aggregate burstiness once with Theorem~\ref{thm:hetero} and once with Theorem~\ref{thm:semi_homog}. When groups are small (here of 10 flows),  Theorem~\ref{thm:semi_homog} provides better bounds than Theorem~\ref{thm:hetero}, but when groups are larger (here of 20 flows), Theorem~\ref{thm:hetero} dominates Theorem~\ref{thm:semi_homog}.

\section{Conclusion}\label{sec:conclusion}
 In this paper, we provided quasi-deterministic bounds on the aggregate burstiness for independent, periodic flows. When a small violation tolerance, is allowed, the bounds are considerably better compared to the deterministic bounds.  We obtained a closed-form expression for the homogeneous case, and for the heterogeneous case, we combined bounds obtained for homogeneous sets using the convolution bounding technique. 
 
 We on purpose limited our study to the burstiness. Quasi-deterministic delay and backlog bounds can be obtained by applying any method from  deterministic network calculus, and combining, either by mean of the union bound or (in case of independence) convolution-like manipulations of the burstiness violation events defined for this paper for all groups of flows. Our results can for example be directly applied to~\cite[Theorem 5]{BN15}, where the model $\sbbThree$ was used to compute probabilistic delay bounds in tandem networks. 

\bibliographystyle{splncs04}
\bibliography{leb,ref}

\end{document}